\documentclass{amsart}[13pt]
\usepackage{amsmath}
\usepackage{graphicx}
\usepackage{graphicx,xcolor}
\usepackage{amsfonts} 
\usepackage{amssymb}
\usepackage{pdfsync}
\usepackage{color}

\usepackage{epsfig}
\theoremstyle{plain}
\newtheorem{theorem}{Theorem}[section]

\newtheorem{lemma}[theorem]{Lemma}
\newtheorem{assumption}{Assumption}
\newtheorem*{assumption*}{Assumption}
\newtheorem{proposition}[theorem]{Proposition}

\newtheorem{remark}[theorem]{Remark}

\theoremstyle{definition}
\newtheorem*{remark*}{Remark}

\theoremstyle{remark}
\numberwithin{equation}{section}

\newcommand{\D}{\mathcal{D}}

\newcommand{\ep}{\varepsilon}

\newcommand{\ffi}{\varphi}

\newcommand{\R}{\mathbb{R}}

\newcommand{\PP}{\mathbb{P}}
\newcommand{\EE}{\mathbb{E}}

\newcommand{\T}{\mathcal{T}}
\newcommand{\LL}{\mathcal{L}}

\newcommand{\mbf}{\mathbf}

\newcommand{\h}{\tilde h_\ep}

\usepackage{latexsym}

\title[Derivation of the Fick's law.]{Derivation of the Fick's law for the Lorentz model in a low density regime}
\author[G. Basile]
{G. Basile}
\address[Giada Basile]{Dipartimento di Matematica ``Guido Castelnuovo'', Sapienza Universit\`a di Roma, P.le Aldo Moro 5, I-00185 Roma, Italy}
\email[G. Basile]{basile@mat.uniroma1.it}
\author[A. Nota]
{A. Nota}
\address[Alessia Nota]{Dipartimento di Matematica ``Guido Castelnuovo'', Sapienza Universit\`a di Roma, P.le Aldo Moro 5, I-00185 Roma, Italy}
\email[A. Nota]{nota@mat.uniroma1.it}
\author[F. Pezzotti]
{F. Pezzotti}
\address[Federica Pezzotti]{Dipartimento di Matematica ``Guido Castelnuovo'', Sapienza Universit\`a di Roma, P.le Aldo Moro 5, I-00185 Roma, Italy}
\email[F. Pezzotti]{pezzotti@mat.uniroma1.it}
\author[M. Pulvirenti]
{M. Pulvirenti}
\address[Mario Pulvirenti]{Dipartimento di Matematica ``Guido Castelnuovo'', Sapienza Universit\`a di Roma, P.le Aldo Moro 5, I-00185 Roma, Italy}
\email[M. Pulvirenti]{pulvirenti@mat.uniroma1.it}

\begin{document}
\begin{abstract}
We consider the Lorentz model in a slab with two mass reservoirs at the boundaries. 
 We show that, in a low density regime, there exists a unique stationary solution for the microscopic dynamics which converges to the stationary solution of the heat equation, namely to the linear profile of the density. In the same regime the macroscopic current in the stationary state is given by the Fick's law, with the diffusion coefficient determined by the Green-Kubo formula.

\end{abstract}
\maketitle

\tableofcontents

\section{Introduction}
One of the most important and challenging problem in the rigorous approach to non-equilibrium Statistical Mechanics is the characterization of stationary nonequilibrium states exhibiting transport  phenomena such as energy or mass transport,  which are macroscopically described by Fourier's  and Fick's law respectively.
A simple microscopic model to validate the Fick's Law is the Lorentz gas, namely a system of non interacting light particles in a distribution of scatterers, in contact with two mass reservoirs. One expects that under a suitable space-time scaling (hydrodynamical limit) the stationary mass current is proportional to the gradient of the density. However the rigorous proof of that is a difficult and still open problem. 
%The more ambitious problem of detecting the Fick's law for a dense scatterer distribution, in a diffusive scaling, eludes our techniques and it seems to be a far target.

%fai esempio del gas di lorentz denso come il pi semplice modello per ottenere la legge di fick. Troppo difficile il caso denso. unico risultato  bunimovich sinai nel caso libero non fuori equilibrio. e quindi lega alla frase dopo tagliando "unfortunatelyÉ." 
%The only existing result is the convergence of the steady state to the not of boltzmann . the contribution gives a modest but nontrivial step towards a challenging goal.

%Unfortunately, there are very few rigorous results on these arguments in the current literature.  

In this paper we propose a contribution in this direction %by validating the Fick's law for the Lorentz model, which is a system of non interacting light particles in a random distribution of scatterers, 
in a situation of low density.
%Il fatto di farlo in bassa densitˆ  per una approssimazione markovianaÉ.metti la parte bridge boa bla
The system we study is the following. Consider the two-dimensional strip
$
\Lambda= (0,L)  \times \R .
$
In the left and in the right of the boundaries,  $ \{0\}\times \R$ and $\{L\}\times \R$ respectively, there are two mass reservoirs constituted by free point particles at equilibrium at different densities $\rho_1$, $\rho_2$.
Inside the strip there is a random distribution of hard disks of radius $\ep$, distributed according to a Poisson law with density $\mu_{\ep}$. Here $\ep$ is a small scale parameter and we let it go to zero.  In the mean time 
$\mu_{\ep}$ is diverging in such a way that  $\mu_{\ep} \ep \to \infty $ and $\mu_{\ep} \ep^2 \to 0 $. Therefore the scatterer configuration is dilute. 

The light particles are flowing through the boundaries, from right with density $\rho_2$ and from left with density $\rho_1$ %(with $\rho_2>\rho_1$). 
They are not interacting among themselves, but are elastically reflected by the obstacles. %Since the elastic collisions preserve the kinetic energy of the light particles, we can assume they are distributed in the unit circle $S_1$.  In particular, in order to have a dilute configuration of scatterers, 
Their mean free paths vanish as $\ep\to 0$, but not too quickly. More precisely they can vanish at most as $\ep^{1-\delta},\, 0<\delta<1,$ in order to have a dilute configuration of scatterers.

We expect that there exists a stationary state for which
\begin{equation}
\label{FL}
J \approx -D \nabla \rho 
\end{equation}
where $J$ is the mass current, $\rho$ is the mass density and $D>0$ is the diffusion coefficient. Formula \eqref{FL} is the well known Fick's law which we want to prove in the present context.

We underline preliminary that our result holds in a low-density regime. This means that we can use the linear Boltzmann equation as a bridge between our original mechanical system and the diffusion equation. This basic idea has been used in  \cite{ESY} \cite{BGS-R} \cite{BNP} to obtain the heat equation from a particle system in different contexts. It works once having an explicit control of the error in the kinetic limit, which suggests the scale of times for which the diffusive limit can be achieved. 
As a consequence the diffusion coefficient $D$ is given by the Green-Kubo formula for the kinetic equation at hand (namely linear Quantum Boltzmann for \cite{ESY}, linear Boltzmann for \cite{BGS-R}, linear Landau for \cite{BNP}).
In the present paper we work in a stationary situation for which we face new problems which will be discussed later on. 

%The more ambitious problem of detecting the Fick's law for a dense scatterer distribution, in a diffusive scaling, eludes our techniques and it seems to be a far target.

The idea of using the linear Boltzmann equation for the Lorentz gas in not new. In \cite{LS} the authors consider exactly our system but with two thermal reservoirs at different temperatures at the boundaries. The aim was to study the energy flux in a stationary regime. However, as pointed out in \cite{LS}, due to the energy conservation of a single elastic collision, the energy is not diffused, there is no local equilibrium and hence the local temperature is not defined. As a consequence the Fourier's law fails to hold, at least in the conventional sense.

This is the reason why we consider here the mass transport, being the heat equation for the mass density the unique hydrodynamical equation. 

It may be worth to mention that, for a suitable stochastic dynamics, the Fourier's law can indeed be derived, see \cite{KMP}, \cite{GKMP}. 

Concerning the Fick's law we mention the papers \cite{LS1}, \cite{LS2}, for the self-diffusion of a tagged particle in a gas at equilibrium.

Our paper is organized as follows.
The starting point is the transition from the mechanical system to the Boltzmann equation in a low density regime. We follow the classical analysis due to Gallavotti \cite{G}, complemented  by an explicit analysis of the bad events preventing the Markovianity, in the same spirit of \cite{DP}, \cite{DR} . This is necessary  to reach a diffusive behavior on a longer time scale as in \cite{BGS-R}, \cite{BNP}.

Moreover we point out that our  initial boundary value problem presents a new feature due to the presence of the first exit  (stopping)  time. This difficulty is handled by an extension procedure which  essentially reduces our problem to the corresponding one in the whole space.

The transition from the mechanical system to the linear Boltzmann regime is presented in Section \ref{sec5}.

However we are interested in a stationary problem. This  is handled, more conveniently, in terms of a Neumann series %rather than as the limit for large time, 
to overcome problems connected with the exchange of the limits $t \to \infty$, $\ep \to 0$. To the best of our knowledge this is a new tool. 
This analysis is presented in Section \ref{proofs}. The basic idea is that the explicit solution of the heat equation and the control of the time dependent problem allow us to characterize the stationary solution of the linear Boltzmann equation and this turns out  to be the basic tool to obtain the stationary solution of the mechanical system which is the basic object of our investigation.

Finally the transition from Boltzmann to the diffusion equation is classical and ruled out by the Hilbert  expansion method which is presented in Section \ref{sec4HILBERT}.
This step is discussed in detail, not only for completeness, but also because we need an apparently new analysis in $L^\infty$, for the time dependent problem (needed for the control of the Neumann series) and a $L^2$ analysis for the stationary problem.

\section{The model and main results}\label{sec2}
Let $\Lambda \subset \R^2$ be the strip $(0,L)\times \R$.
We consider a Poisson distribution of fixed hard disks (scatterers) of radius $\ep$ in $\Lambda$ and denote by $c_1,\dots,c_N\in\Lambda$ their centers. 
This means that, given $\mu>0$, the probability density of finding $N$ obstacles in a bounded measurable set $A\subset\Lambda$ is 
\begin{equation}\label{poisson}
\PP(\,d\mbf{c}_{N})=e^{-\mu |A|}\frac{\mu^N}{N!}\,dc_1\dots\,dc_N
\end{equation}
where $|A|=\text{meas}A$ and $\mbf{c}_{N}=(c_1,\dots, c_N)$.\\

A particle in $\Lambda$ moves freely up to the first instant of contact with an obstacle. Then it is elastically reflected and so on. Since the modulus of the velocity of the test particle is constant, we assume it to be equal to one, so that the phase space of our system is $\Lambda\times S_1$.

We rescale the intensity $\mu$ of the obstacles as 
$$
\mu_{\varepsilon}=\varepsilon^{-1}\eta_{\ep}\mu,
$$
where, from now on, $\mu>0$ is fixed and $\eta_{\ep}$ is slowly diverging as $\ep\to 0$. More precisely we make the following assumption.
\begin{assumption}\label{A1}
As $\varepsilon \to 0$, $\eta_\ep$ diverges in such a way that 
\begin{equation}\label{assump}
\ep^{\frac{1}{2}}\eta_\ep^6\to 0.
\end{equation}
\end{assumption}

The behaviour \eqref{assump} is dictated mostly by the recollision estimates in Section \ref{pathological}.

We denote by $\PP_{\ep}$  the probability density \eqref{poisson} with $\mu$ replaced by $\mu_\ep$. $\EE_{\ep}$ will be the expectation with respect to the measure $\PP_{\ep}$ restricted on those configurations of the obstacles whose centers do not belong to the disk of center $x$ and radius $\ep$. 

For a given configuration of obstacles $\mbf{c}_N$, we denote by $T^{-t}_{\mbf{c}_{N}}(x,v)$ the (backward) 
flow with initial datum $(x,v)\in\Lambda\times S_1$ and define
$t-\tau$, $\tau=\tau(x,v,t,\mbf{c}_N)$, as the first (backward) hitting time with the boundary. We use the notation $\tau=0$ to indicate the event such that the trajectory $T^{-s}_{\mbf{c}_{N}}(x,v)$, $s\in [0,t]$, never hits the boundary.
For any $t\geq 0$ the one-particle correlation function reads
\begin{equation}\label{def:fep}
f_{\ep}(x,v,t)=\EE_\varepsilon[f_B (T^{-(t-\tau)}_{\mbf{c}_{N}}(x,v))\chi(\tau>0)] + \EE_\ep[f_0 (T^{-t}_{\mbf{c}_{N}}(x,v))\chi(\tau=0)],
\end{equation} 
where $f_0\in L^\infty(\Lambda\times S_1)$ and the boundary value $f_B$ is defined by
\begin{equation*}
f_B(x,v):=\left\{\begin{array}{ll}
\rho_1 M(v)\quad\text{if}\quad x\in \{0\}\times\R,\quad v_1>0,&\vspace{3mm} \\
\rho_2 M(v)\quad\text{if}\quad x\in \{L\}\times\R,\quad v_1<0,& \vspace{3mm}
\end{array}\right.
\end{equation*}
with $M(v)$ the density of the uniform distribution on $S_1$ and $\rho_1, \rho_2>0$.  Here $v_1$ denotes the horizontal component of the velocity $v$.
Without loss of generality we assume $\rho_2>\rho_1$.
Since $M(v)=\frac{1}{2\pi}$, from now on we will absorb it in the definition of the boundary values $\rho_1, \rho_2$. Therefore we set
\begin{equation}\label{def:fB}
f_B(x,v):=\left\{\begin{array}{ll}
\rho_1\quad\text{if}\quad x\in \{0\}\times\R,\quad v_1>0,&\vspace{3mm} \\
\rho_2\quad\text{if}\quad x\in \{L\}\times\R,\quad v_1<0.& \vspace{3mm}
\end{array}\right.
\end{equation}

\begin{remark*} Here we allow overlapping of scatterers, namely the Poisson measure is that of a free gas. It would also be possible to consider the Poisson measure restricted to non-overlapping configurations, namely the Gibbs measure for a systems of hard disks in the plane. However the two measures are asymptotically equivalent and the result does hold also in the last case.

Note also that the dynamics $T^{t}_{\mbf{c}_{N}}$  is well defined only almost everywhere with respect to $\PP_{\ep}$. %to avoid trajectories of the light particle hitting corners. 
\end{remark*}

We are interested in the stationary solutions $f_{\ep}^S$ of the above problem. More precisely for any $t\geq 0$ $f_\ep^S(x,v)$ solves  
\begin{equation}\label{def:ST}
f_{\ep}^S(x,v)=\EE_\varepsilon[f_B (T^{-(t-\tau)}_{\mbf{c}_{N}}(x,v))\chi(\tau>0)] + \EE_\ep[f_ \ep^S(T^{-t}_{\mbf{c}_{N}}(x,v))\chi(\tau=0)].
\end{equation} 

The main result of the present paper can be summarized in the following theorem.

\begin{theorem}\label{th:MAIN1}
For $\ep$ sufficiently small there exists a unique $L^\infty$ stationary solution $f_\ep^S$ for the microscopic dynamics (i.e. satisfying \eqref{def:ST}). Moreover, as $\varepsilon \to 0$ 
\begin{equation}\label{convergenzaSTAZ}
f_\ep^S\rightarrow \varrho^S,
\end{equation}
where $\varrho^S$ is the stationary solution of the heat equation with the following boundary conditions
\begin{equation}\label{HEAT}
\left\{
 \begin{array}{l}\vspace{0.2cm}
\varrho^S(x)=\rho_1,\ \ \ \ \ x\in \{0\}\times\R,  
\vspace{0.2cm}\\
\varrho^S(x)=\rho_2,\ \ \ \ \ x\in \{L\}\times\R.
\end{array} \right.
\end{equation}
The convergence is in $L^{2}((0,L)\times S_1)$.
\end{theorem}

Some remarks on the above Theorem are in order. 
The boundary conditions of the problem
depend on the space variable only through the horizontal component. As a consequence, the stationary solution $f_\ep^S$ of the microscopic problem, as well as the stationary solution $\varrho^S$ of the heat equation, inherits the same feature. This justifies the convergence in $L^{2}((0,L)\times S_1)$ instead of in $L^{2}(\Lambda\times S_1)$. 
The explicit expression for the stationary solution $\varrho^S$ reads 
\begin{equation}\label{statheat}
\varrho^S(x)= \frac{\rho_1(L-x_1)+\rho_2 x_1}{ L},
\end{equation}
where $x_1$ is the horizontal component of the space variable $x$. We note that in order to prove Theorem \ref{th:MAIN1} it is enough to assume that $\ep^{\frac 1 2}\eta_{\ep}^5\to 0.$ The stronger Assumption \ref{A1} is needed to prove Theorem \ref{th:MAIN2} below.\\\vspace{2mm}

Next we discuss the Fick's law by introducing the stationary mass flux 
\begin{equation}\label{def:j}
J_\ep^S(x)=\eta_{\ep}\int_{S_1}v\,f_\ep^S(x,v)\,dv,
\end{equation}
and the stationary mass density
\begin{equation}\label{def:mass}
\varrho_\ep^S(x)=\int_{S_1}f_\ep^S(x,v)\,dv.
\end{equation}
Note that $J_\ep^S$ is the total amount of mass flowing through a unit area in a unit time interval. Although in a stationary problem there is no typical time scale, the factor $\eta_{\ep}$ appearing in the definition of $J_\ep^S$, is reminiscent of the time scaling necessary to obtain a diffusive limit.

%%%%%%%%%
\begin{theorem}[\textbf{Fick's law}] \label{th:MAIN2}
We have 
\begin{equation}\label{FickL}
J_\ep^S+D\nabla_x\varrho_\ep^S\to 0
\end{equation}
as $\ep\to 0$. The convergence is in $\D'(0,L)$ and $D>0$ is given by the Green-Kubo formula (see \eqref{GK} below).
Moreover
\begin{equation}\label{Jlim}
J^S=\lim_{\ep\to 0}J_\ep^S(x),
\end{equation}
where the convergence is in $L^2(0,L)$
and 
\begin{equation}\label{FICK}
J^S=-D\,\nabla\varrho^S=-D\,\frac{\rho_2-\rho_1}{L},
\end{equation}
where $\varrho^S$ is the linear profile \eqref{statheat}.
\end{theorem}
\vspace{4mm}
Observe that, as expected by physical arguments, the stationary flux $J^S$ does not depend on the space variable. Furthermore the diffusion coefficient $D$ is determined by the behavior of the system at equilibrium and in particular it is equal to the diffusion coefficient for the time dependent problem.

\begin{remark*}[The scaling]
We have formulated our result in macroscopic variables $x,t$.
Another point of view is to argue in terms of microscopic variables.

Let us set our problem in these variables denoted by $(q,t')$. This means that the radius of the disks is unitary while the strip, seen in micro-variables, is $(0,{\ep}^{-1}L)\times \R$.

To deal with a low density situation, we rescale the density as $\eta_{\ep}\ep\mu$, $\mu>0$ where $\eta_{\ep}$ is gently diverging. Note that in the usual Boltzmann-Grad limit $\eta_{\ep}=1.$ 
At times of order $\ep^{-1}$, one particle has an average number of collisions of order $\eta_{\ep}$. At larger times, namely of order $\eta_{\ep}\ep^{-1}$, we expect a diffusive behavior. Actually this emerges from the linear Boltzmann equation (see equation \eqref{eq:Bfree} and Proposition \ref{prop2} below) which is derived from the microscopic dynamics through the scaling $x=\ep q$ and $t=\ep\eta_{\ep}^{-1}t'$.
\end{remark*}

%AGGIUNTA 
In this paper we consider a two dimensional case but our techniques apply in higher dimensions as well since in this case the pathological events are less likely. Moreover we consider the easier geometrical setting. However we believe that there are no serious obstructions to extend our results to more general geometries.

\section{Proofs}\label{proofs}
In this section we prove Theorems \ref{th:MAIN1} and \ref{th:MAIN2}, postponing the technical details to the next sections. In order to prove Theorem \ref{th:MAIN1} our strategy is the following. 
We introduce the stationary linear Boltzmann equation 
\begin{equation}\label{eq:Boltz2}
\left\{\begin{array}{ll}
\big(v\cdot\nabla_x\big)h_\ep^S(x,v)=\eta_\ep\, \LL h_\ep^S(x,v),
&\vspace{0.2cm}\\
h_\ep^S(x,v)=\rho_1,\ \ \ \ \ x\in \{0\}\times \R,\quad v_1>0, 
&\vspace{0.2cm}\\
h_\ep^S(x,v)=\rho_2,\ \ \ \ \ x\in \{L\}\times \R,\quad v_1<0,\  
&
\end{array}\right.
\end{equation}
where $\LL$ is the linear Boltzmann operator defined as 
\begin{equation}\label{def:L_ve}
 \LL f (v)=\mu\int_{-1}^1d\rho\big[f(v')-f(v)\big],\qquad  f\in L^1(S_1)
\end{equation}
with
\begin{equation}\label{scattering}
v'=v-2(n\cdot v)n
\end{equation}
and $n=n(\rho)$ the outward normal to the hard disk (see Figure 
\ref{fig:opK}). Here $\rho$ is the impact parameter, namely $\rho=\sin\alpha$ with $\alpha$ the angle of incidence.
%FIGURA
\begin{figure}[ht]
\centering
\includegraphics[scale= 0.2]{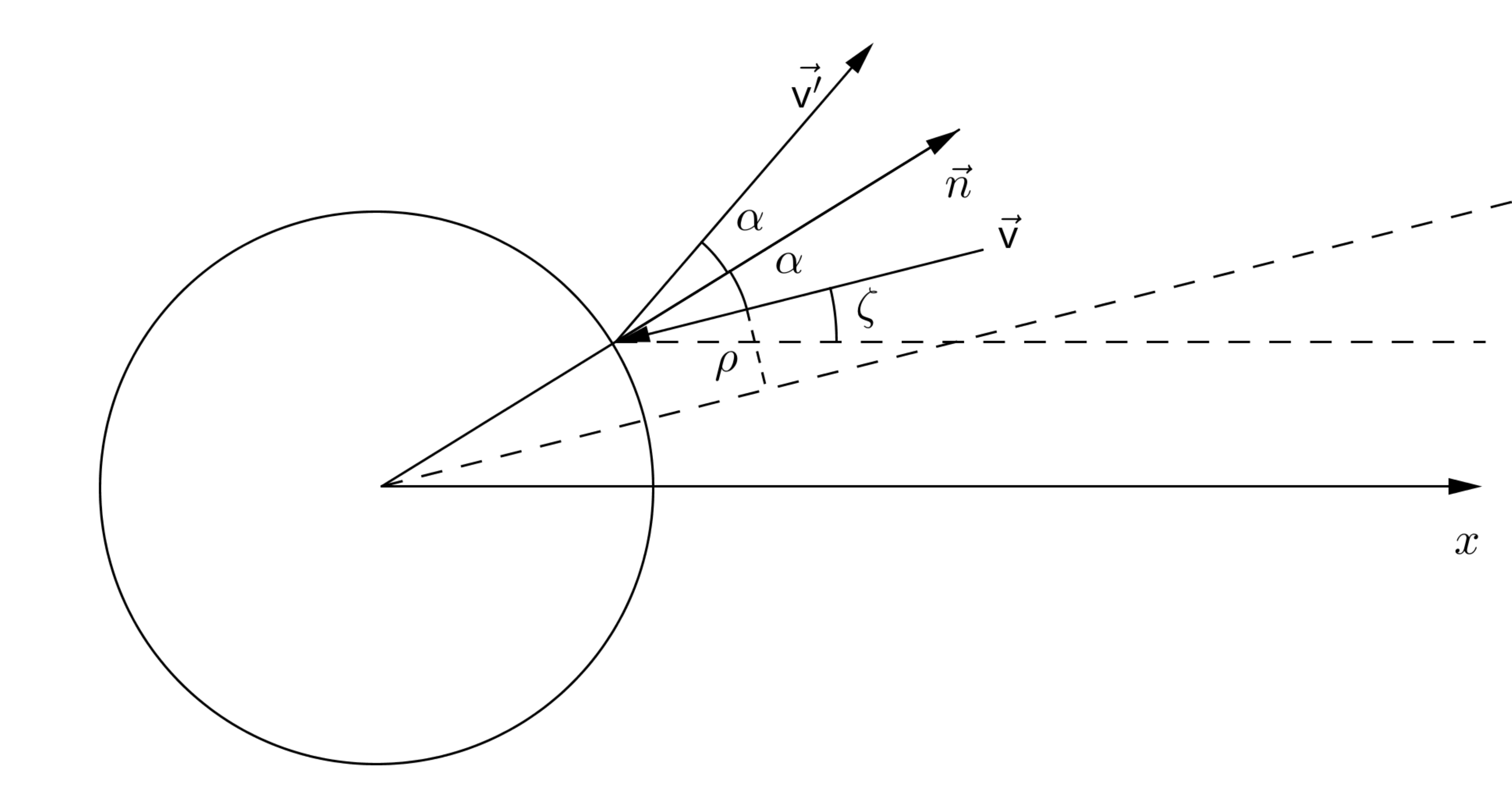}
\caption{Elastic reflection}
\label{fig:opK}
\end{figure}

Since the boundary conditions depend on the space variable only through the horizontal component, the stationary solution $h_\ep^S$ inherits the same feature, as well as $f_\ep^S$ and $\varrho^{S}$. 

The strategy of the proof consists of two steps. First we prove that there exists a unique $h_\ep^S$ which converges, as $\ep\to 0$, to $\varrho^S$ given by \eqref{statheat}. See Proposition \ref{prop:hSrhoS} below. Secondly we show that there exists a unique $f_{\ep}^S$ asymptotically equivalent to $h_\ep^S$. See Proposition \ref{prop:fShS} below. This result is achieved by showing that the memory effects of the mechanical system, preventing the Markovianity, are indeed negligible.

\vspace{2mm}
Let $h_\ep$ be the solution of the  
problem
\begin{equation}\label{eq:Boltz1}
\left\{\begin{array}{ll}
\big(\partial_t +v\cdot\nabla_x\big)h_\ep(x,v,t)=\eta_\ep\,  \LL h_\ep(x,v,t),
&\vspace{0.2cm}\\
h_\ep(x,v,0)=f_0(x,v), \ \ \ \ \ \ f_0\in L^{\infty}(\Lambda\times S_1),&\vspace{0.2cm}\\
h_\ep(x,v,t)=\rho_1,\ \ \ \ \ x\in \{0\}\times \R,\quad v_1>0,\ \ \ t\geq 0,\  
&\vspace{0.2cm}\\
h_\ep(x,v,t)=\rho_2,\ \ \ \ \ x\in \{L\}\times \R,\quad v_1<0,\ \ \ t\geq 0.\  
&
\end{array}\right.
\end{equation}

Then $h_\ep$ has the following explicit representation
\begin{equation}
\begin{split}
\label{eq:heps}
h_{\ep}(x,v,t)&= \sum_{N\geq 0} \left(\mu_\ep\ep\right)^{N}\int_{0}^{t}dt_1\dots\int_{0}^{t_{N-1}}dt_N\\&
\int_{-1}^{1}d\rho_1\dots\int_{-1}^{1}d\rho_N \, \chi(\tau<t_N)\chi(\tau>0)\, e^{-2\mu_\ep\ep \, (t-\tau)}\, f_{B}(\gamma^{-(t-\tau)}(x,v))+\\&
+\sum_{N\geq 0} e^{-2\mu_\ep\ep \, t}\left(\mu_\ep\ep\right)^{N}\int_{0}^{t}dt_1\dots\int_{0}^{t_{N-1}}dt_N\\&
\int_{-1}^{1}d\rho_1\dots\int_{-1}^{1}d\rho_N\, \chi(\tau=0)\, f_0(\gamma^{-t}(x,v)),
\end{split}
\end{equation}
with $f_B$ defined in \eqref{def:fB}. 
Given $x,v,\,t_1\dots t_N,\,\rho_1\dots\rho_N$, $\gamma^{-t}(x,v)$ denotes the trajectory whose position and velocity are
$$(x-v(t-t_1)-v_1(t_1-t_2)\dots-v_Nt_N,v_N).$$
The transitions $v\to v_1\to v_2\dots\to v_N$ are obtained by means of a scattering with an hard disk with impact parameter $\rho_i$ via \eqref{scattering}.
As before $t-\tau$, $\tau=\tau(x,v,t_1\dots,t_N,\rho_1\dots \rho_N)$, is the first (backward) hitting time with the boundary. We remind that $\mu_{\ep}\ep=\mu\eta_{\ep}.$ 

In formula \eqref{eq:heps} $h_{\ep}(t)$ results as the sum of two contributions, one due to the backward trajectories hitting the boundary and the other one due to the trajectories which never leave $\Lambda$.
Therefore we set
\begin{equation*}\label{def:hep}
h_{\ep}(x,v,t)=h_{\ep}^{out}(x,v,t)+h_{\ep}^{in}(x,v,t),
\end{equation*} 
where $h_{\ep}^{out}$ and $h_{\ep}^{in}$ are respectively the first and the second sum on the right hand side of \eqref{eq:heps}. Observe that
$h_{\ep}^{out}$ solves 
\begin{equation}\label{eq:Boltz}
\left\{\begin{array}{ll}
\big(\partial_t +v\cdot\nabla_x\big)h_\ep^{out}(x,v,t)=\eta_\ep\, \LL h_\ep^{out}(x,v,t),
&\vspace{0.2cm}\\
h_\ep^{out}(x,v,0)=0, \ \ \ \ \ \ x\in \Lambda,&\vspace{0.2cm}\\
h_\ep^{out}(x,v,t)=\rho_1,\ \ \ \ \ x\in \{0\}\times \R,\quad v_1>0,\ \ \ t\geq 0,\  
&\vspace{0.2cm}\\
h_\ep^{out}(x,v,t)=\rho_2,\ \ \ \ \ x\in \{L\}\times \R,\quad v_1<0,\ \ \ t\geq 0.\  
&
\end{array}\right.
\end{equation}
We denote by $S_\ep^0(t)$ the Markov semigroup associated to the second sum, namely 
\begin{equation*}\label{def:S0}
\begin{split}
(S_\ep^0(t)\ell)(x,v)&=\sum_{N\geq 0} e^{-2\mu_\ep\ep \, t}\left(\mu_\ep\ep\right)^{N}\int_{0}^{t}dt_1\dots\int_{0}^{t_{N-1}}dt_N\\&
\int_{-1}^{1}d\rho_1\dots\int_{-1}^{1}d\rho_N\, \chi(\tau=0)\,\ell(\gamma^{-t}(x,v)),
\end{split}
\end{equation*} 
with $\ell\in L^{\infty}(\Lambda\times S_1).$
In particular
$$
h_{\ep}^{in}(t)=S_\ep^0(t)f_0.
$$ 
We observe that $h_\ep^S$, solution of \eqref{eq:Boltz2}, satisfies, for $t_0>0$
\begin{equation*}%\label{eq:hS}
h_\ep^S=h_{\ep}^{out}(t_0)+S_\ep^0(t_0)h_\ep^S,
\end{equation*}
so that we can formally express $h_\ep^S$ as the Neumann series
\begin{equation}\label{eq:hSN}
h_\ep^S=\sum_{n\geq 0}(S_\ep^0(t_0))^n h_\ep^{out}(t_0).
\end{equation}

\begin{remark*}
Note that $h_\ep^S$ is a fixed point of the map $f_0\to h_\ep(t_0)$ solution to \eqref{eq:Boltz1}. Hence $h_\ep^S$ belongs to a periodic orbit, of period $t_0$, of the flow $f_0\to h_\ep(t)$. But this orbit consists of a single point because the Neumann series, being convergent, identifies a single element. This implies that $h_\ep^S$ is constant with respect to the flow \eqref{eq:Boltz1} and hence stationary.
\end{remark*}
We now establish existence and uniqueness of $h_\ep^S$ by showing 
that the Neumann series \eqref{eq:hSN} converges. In order to do it we need to extend the action of the semigroup $S_\ep^0(t)$ to the space $L^\infty(\R^2\times S_1)$, namely 
\begin{equation}
\begin{split}
\label{def:S01}
S_\ep^0(t)\ell_0(x,v)&=\chi_{\Lambda}(x)\sum_{N\geq 0} e^{-2\mu_\ep\ep \,t}\left(\mu_\ep\ep\right)^{N}\int_{0}^{t}dt_1\dots\int_{0}^{t_{N-1}}dt_N\\&
\int_{-1}^{1}d\rho_1\dots\int_{-1}^{1}d\rho_N\, \chi(\tau=0)\, \ell_0(\gamma^{-t}(x,v)),
\end{split}
\end{equation}
for any $\ell_0(x,v)\in L^\infty(\R^2\times S_1).$ Here $\chi_{\Lambda}$ is the characteristic function of $\Lambda$. 

\begin{proposition}\label{prop:existencehepS}
There exists $\ep_0>0$ 
such that for any $\ep<\ep_0$ and for any $\ell_0\in L^\infty(\R^2\times S_1)$ we have
\begin{equation}\label{S0bounded}
||S_\ep^0(\eta_{\ep})\ell_0 ||_\infty\leq \alpha\, ||\ell_0||_\infty,\ \ \ \alpha<1.
\end{equation}
As a consequence there exists a unique stationary solution $h_\ep^S\in L^\infty (\Lambda\times S_1)$ satisfying (\ref{eq:Boltz2}).
\end{proposition}
To prove Proposition \ref{prop:existencehepS} we have first to exploit the diffusive limit of the linear Boltzmann equation in a $L^{\infty}$ setting and in the whole space. 
We introduce $\h:\R^2\times S_1\times [0,T]\to \R^{+}$ the solution of the following rescaled linear Boltzmann equation
\begin{equation}\label{eq:Bfree}
\left\{
\begin{array}{l}\vspace{4mm}
\big(\partial_t  +\eta_\ep\,v\cdot\nabla_x\big)\h=\eta_\ep^2\,  \LL \h\\
\h(x,v,0)=\varrho_0(x),
\end{array}
\right.
\end{equation}
with $\varrho_0$ is a smooth function of the variable $x$ only (local equilibrium). 

We can prove
\begin{proposition}\label{prop2}
Let $\h$ be the solution of \eqref{eq:Bfree}, with an initial datum $\varrho_0\in C^\infty_0(\R^2)$.
Then, as $\ep\to 0$, $\h$ converges to the solution of the heat equation
\begin{equation}\label{eq:diff0}
\left\{
\begin{array}{l}\vspace{0.4cm}
\partial_t \varrho-D\Delta\varrho=0\\
\varrho(x,0)=\varrho_0(x),
\end{array}
\right.
\end{equation}
where $D$ is given by the Green-Kubo formula
\begin{equation}\label{GK}
D=\frac 1 {4\pi}\int_{S_1}dv\, v\cdot\big(-\LL\big)^{-1}v. 
\end{equation}
The convergence is in $L^{\infty}([0,T];L^{\infty}(\R^2\times S_1)).$
\end{proposition}
We postpone the proof of Proposition \eqref{prop2} to Section \ref{sec:HLinf}.
The proof relies on the Hilbert expansion and, to make it work, we need smoothness of the initial datum $\varrho_0$.

\begin{proof} [Proof of Proposition \ref{prop:existencehepS}] 
We can rewrite \eqref{def:S01} as
\begin{equation*}
\begin{split}
%\label{def:S02}
S_\ep^0(t)\ell_0(x,v)&=\chi_{\Lambda}(x)\sum_{N\geq 0} e^{-2\mu_\ep\ep \, t}\left(\mu_\ep\ep\right)^{N}\int_{0}^{t}dt_1\dots\int_{0}^{t_{N-1}}dt_N\\&
\int_{-1}^{1}d\rho_1\dots\int_{-1}^{1}d\rho_N\, \chi(\tau=0)\, \ell_0(\gamma^{-t}(x,v))\chi_{\Lambda}(\gamma_v^{-t}(x)),
\end{split}
\end{equation*}
where $\gamma_v^{-t}(x)=x-v(t-t_1)-v_1(t_1-t_2)\dots-v_Nt_N$ is the first component of $\gamma^{-t}(x,v)$. Note that the insertion of $\chi_{\Lambda}(\gamma_v^{-t}(x))$ is due to the constraint $\chi(\tau=0)$. 
Therefore
\begin{equation*}
\begin{split}
S_\ep^0(t)\ell_0&\leq ||\ell_0||_\infty \sum_{N\geq 0} e^{-2\mu_\ep\ep \, t}\left(\mu_\ep\ep\right)^{N}\int_{0}^{t}dt_1\dots\int_{0}^{t_{N-1}}dt_N\\&
\int_{-1}^{1}d\rho_1\dots\int_{-1}^{1}d\rho_N\, \chi_{\Lambda}(\gamma_v^{-t}(x)).
\end{split}
\end{equation*}
We denote by $\chi^{\delta}_{\Lambda}$ a mollified version of $\chi_{\Lambda}$, namely $\chi^{\delta}_{\Lambda}\in C_0^{\infty}(\R^2)$, $\chi^{\delta}_{\Lambda}(x)\leq 1$, $\chi^{\delta}_{\Lambda}\geq \chi_{\Lambda}$ and $\textit{supp}(\chi^{\delta}_{\Lambda})\subset (-\delta, L+\delta)\times \R$. Therefore
\begin{equation}\label{estS0}
\begin{split}
S_\ep^0(t)\ell_0&\leq ||\ell_0||_\infty \sum_{N\geq 0} e^{-2\mu_\ep\ep \, t}\left(\mu_\ep\ep\right)^{N}\int_{0}^{t}dt_1\dots\int_{0}^{t_{N-1}}dt_N\\&
\int_{-1}^{1}d\rho_1\dots\int_{-1}^{1}d\rho_N\, 
\chi^{\delta}_{\Lambda}(\gamma_v^{-t}(x)).
\end{split}
\end{equation}
The series on the right hand side of \eqref{estS0} defines a function $F$ which solves
\begin{equation*}
%\label{BoltzmannNONrisc}
\left\{\begin{array}{ll}
(\partial_{t}+ v\cdot\nabla_x)
F(x,v,t)=\eta_{\ep} \LL 
F(x,v,t),& \vspace{3mm}\\
F(x,v,0)=\chi^{\delta}_{\Lambda}(x).
\end{array}\right.
\end{equation*}
Moreover, defining $G_\ep(x,v,t):=F(x,v,\eta_\ep t)$ then $G_\ep$ solves \eqref{eq:Bfree} with initial datum $\varrho_0=\chi^{\delta}_{\Lambda}$. 
By virtue of Proposition \ref{prop2}
$$
\|G_\ep(1)-\varrho^{\delta}(1)\|_{\infty}\leq \omega(\ep)
$$
where $\varrho^{\delta}(t)$ is the solution of \eqref{eq:diff0} with initial datum $\chi^{\delta}_{\Lambda}$. Here and in the sequel $\omega(\ep)$ denotes a 
positive function vanishing with $\ep$.
On the other hand
\begin{equation*}
\varrho^{\delta}(x,1)=\int_{\R^2} dy \ \frac{1}{4\pi D} \, e^{-\frac{|x-y|^2}{4D}}\, \chi^{\delta}_{\Lambda}(y)=\int_{-\delta}^{L+\delta} dy_1 \ \frac{1}{\sqrt{4\pi D}} \, e^{-\frac{|x_1-y_1|^2}{4D}}<
1.
\end{equation*}
Therefore for $\ep$ small enough
\begin{equation*}
\begin{split}
||S_\ep^0(\eta_{\ep})\ell_0||_\infty\leq &\, ||\ell_0||_\infty||S_\ep^0(\eta_{\ep})\chi^{\delta}_{\Lambda}||_\infty
\\\leq & \,||\ell_0||_\infty(\|G_\ep(1)-\varrho^{\delta}(1)\|_{\infty}+||\varrho^{\delta}(1)||_\infty)\\ 
\leq &\,||\ell_0||_\infty(\omega(\ep)+||\varrho^{\delta}(1)||_\infty)<\alpha||\ell_0||_\infty,\ \ \ \alpha<1.
\end{split}
\end{equation*}
We are using \eqref{estS0} for $t=\eta_{\ep}$.

Finally, since $\alpha<1$, by \eqref{eq:hSN} we get
\begin{equation*}%\label{hepsnorm}
|| h_\ep^S||_\infty\leq \frac{1}{(1-\alpha)}\ ||h_\ep^{out}(\eta_{\ep})||_\infty\leq \frac{1}{(1-\alpha)}\, \rho_2.
\end{equation*}
\end{proof}
As we will discuss later on,
we find convenient to obtain the stationary solution $h_\ep^S$ via the Neumann series  \eqref{eq:hSN} rather than as the limit of $h_{\ep}(t)$ as $t\to\infty$. For further details see Remark \ref{rk:asymp}.

\begin{remark*}[$L^{\infty}$ vs. $L^2$]
The control of the Neumann series \eqref{eq:hSN} in a $L^{\infty}$ setting seems quite natural. This is provided by the bound \eqref{S0bounded}. It basically means that for a time $\eta_{\ep}$ the probability of a backward trajectory to fall out of $\Lambda$ is strictly positive. 
To prove rigorously this rather intuitive fact, we use Proposition \ref{prop2} and explicit properties of the solution of the heat equation. The price we pay is to develop an $L^{\infty}$ Hilbert expansion analysis (see Section \ref{sec:HLinf}) which is, however, interesting in itself.
On the other hand the use of the well known $L^2$ version of Proposition \ref{prop2} requires a $L^2$ control of the Neumann series which seems harder, weaker and less natural.
\end{remark*}
\vspace{2mm}

The last step is the proof of the convergence of $h_\ep^S$ to the stationary solution of the diffusion problem 
\begin{equation}
\label{HEAT1}
\left\{\begin{array}{ll}
\partial_{t}\varrho-D\Delta\varrho=0
& \vspace{2mm}\\
\varrho(x,t)=\rho_1,\ \ \ \ \ x\in \{0\}\times\R,\ \ \ t\geq 0\  
\vspace{0.2cm}\\
\varrho(x,t)=\rho_2,\ \ \ \ \ x\in \{L\}\times\R,\ \ \ t\geq 0,
\end{array}\right.
\end{equation}
with the diffusion coefficient $D$ given by the Green-Kubo formula \eqref{GK}.
We remind that the stationary solution $\varrho^S$ to the problem \eqref{HEAT1} has the following explicit expression 
\begin{equation}\label{rhoSexp}
\varrho^S(x)= \frac{\rho_1(L-x_1)+\rho_2 x_1}{ L},
\end{equation}
where $x=(x_1,x_2)$. 

By using again the Hilbert expansion technique (this time in $L^2$) we can prove 
\begin{proposition}\label{prop:hSrhoS}
Let $h_\ep^S\in L^{\infty}((0,L)\times S_1)$ be the solution to the problem \eqref{eq:Boltz2}. Then  
\begin{equation}\label{eq:hSrhoS}
h_{\ep}^S\to \varrho^S
\end{equation}
as $\ep\to 0$, where $\varrho^S(x)$ is given by \eqref{rhoSexp}. The convergence is in $L^2((0,L)\times S_1)$.
\end{proposition}
The proof is postponed to Section \ref{sec:HL2}.

This concludes our analysis of the Markov part of the proof. 
\vspace{5mm}

Recalling the expression (\ref{def:fep}) for the one-particle correlation function $f_\ep$, we introduce a decomposition analogous to the one used for $h_{\ep}(t)$, namely
\begin{equation}\label{def:fepOUT}
f_{\ep}^{out}(x,v,t):=\EE_\varepsilon[f_B (T^{-(t-\tau)}_{\mbf{c}_{N}}(x,v))\chi(\tau>0)]
\end{equation} 
and
\begin{equation}\label{def:fepIN}
f_{\ep}^{in}(x,v,t):= \EE_\ep[f_0 (T^{-t}_{\mbf{c}_{N}}(x,v))\chi(\tau=0)],
\end{equation} 
so that 
\begin{equation*}\label{def:fep2}
f_{\ep}(x,v,t)=f_{\ep}^{out}(x,v,t)+f_{\ep}^{in}(x,v,t).
\end{equation*} 
Here $f_{\ep}^{out}$ is the contribution due to the trajectories that do leave $\Lambda$ at times smaller than $t$, while $f_{\ep}^{in}$ is the contribution due to the trajectories that stay internal to $\Lambda$. We introduce the flow $F_\ep^0(t)$ such that 
\begin{equation*}\label{def:F0}
(F_\ep^0(t)\ell)(x,v)=\EE_\ep[\ell (T^{-t}_{\mbf{c}_{N}}(x,v))\chi(\tau=0)], \quad \ell\in L^{\infty}(\Lambda\times S_1)
\end{equation*} 
and remark that $F_\ep^0$ is just the dynamics ''inside'' $\Lambda$. In particular $f_{\ep}^{in}(t)=F_\ep^0(t)f_0.$

To detect the stationary solution $f_\ep^S$ for the microscopic dynamics we proceed as for the Boltzmann evolution (see \eqref{def:ST}) by setting, for $t_0>0$,
\begin{equation*}%\label{eq:fS}
f_\ep^S=f_{\ep}^{out}(t_0)+F_\ep^0(t_0)f_\ep^{S}
\end{equation*}
and we can formally express the stationary solution as the Neumann series
\begin{equation}\label{eq:fSN}
f_\ep^S=\sum_{n\geq 0}(F_\ep^0(t_0))^n f_\ep^{out}(t_0).
\end{equation}
To show the convergence of the series \eqref{eq:fSN} and hence existence of 
$f_\ep^S$ we first need the following two Propositions.

\begin{proposition}\label{th:propCIN}%CINETICO
Let $T>0$. For any $t\in (0,T]$  
\begin{equation}\label{eq:CIN}
\|f_{\ep}^{out}(t)-h_{\ep}^{out}(t)\|_{L^\infty(\Lambda\times S_{1} )}\leq C\varepsilon^{\frac{1}{2}}\, \eta_\ep^3\, t^2,
\end{equation}
where $h_\ep^{out}$ solves (\ref{eq:Boltz}).
\end{proposition}
\begin{proposition}\label{prop:fepinhepin}
For every $\ell_0\in L^\infty(\Lambda\times S_1)$  
 \begin{equation}\label{eq:fepinhepin}
||\left(F_\ep^0(t)-S_\ep^0(t)\right)\ell_0 ||_\infty\leq C||\ell_0||_\infty\,\ep^{\frac 1 2}\eta_\ep^3t^2,\quad \forall t\in[0,T].
\end{equation}
\end{proposition}
The proof of the above two Propositions is postponed to Section \ref{sec5}.
As a corollary we can prove 
\begin{proposition}\label{prop:fShS}
For $\ep$ sufficiently small there exists a unique stationary solution $f_\ep^S\in L^\infty(\Lambda\times S_1)$ satisfying \eqref{def:ST}. 
Moreover 
\begin{equation}\label{eq:fShS}
\|h_\ep^S-f_\ep^S\|_\infty\leq C\ep^{\frac 1 2}\eta_\ep^5.
\end{equation}
\end{proposition}
\begin{proof}
We prove the existence and uniqueness of the stationary solution by showing that the Neumann series \eqref{eq:fSN} converges, namely
\begin{equation}\label{F0bounded}
||F_\ep^0(\eta_{\ep})f_0 ||_\infty\leq \alpha'\, ||f_0||_\infty,\ \ \ \alpha'<1.
\end{equation}
This implies
\begin{equation*}%\label{NeumannEstPARTICLE}
|| f_\ep^S||_\infty\leq \frac{1}{(1-\alpha')}\ ||f_{\ep}^{out}(\eta_{\ep})||_\infty\leq \frac{1}{(1-\alpha')}\, \rho_2,\ \ \ \ \alpha'<1.
\end{equation*}
In fact, since 
\begin{equation*}
||F_\ep^0(\eta_{\ep})f_0 ||_\infty\leq ||\left(F_\ep^0(\eta_{\ep})-S_\ep^0(\eta_{\ep})\right)f_0 ||_\infty+||S_\ep^0(\eta_{\ep})f_0 ||_\infty,
\end{equation*}
thanks to Propositions \ref{prop:existencehepS} and \ref{prop:fepinhepin} we get 
\begin{equation}\label{F0bounded}
\begin{split}
||F_\ep^0(\eta_{\ep})f_0 ||_\infty\leq &\,||f_0||_\infty C\ep^{\frac 1 2}\eta_\ep^5 +  ||S_\ep^0(\eta_{\ep})f_0 ||_\infty\\ \leq &\, (C\ep^{\frac 1 2}\eta_\ep^5+\alpha)||f_0||_\infty
\leq \alpha'||f_0||_\infty, 
\end{split}
\end{equation}
with $\alpha'<1$, for $\ep$ sufficiently small (remind that $ \ep^{\frac 1 2}\eta_{\ep}^5\to 0$ as $\ep\to 0$).
This guarantees the existence and uniqueness of the microscopic stationary solution $f_{\ep}^S$.

In order to prove $\eqref{eq:fShS}$ we compare the two Neumann series representing $f_{\ep}^S$ and $h_{\ep}^S$,
\begin{equation}\label{tbyt}
\begin{split}
\|f^S_{\ep}-h^S_{\ep}\|_{\infty}=&
\,\|\sum_{n\geq 0}\big((F_\ep^0(\eta_{\ep}))^n f_{\ep}^{out}(\eta_{\ep}) -(S_\ep^0(\eta_{\ep}))^n h_{\ep}^{out}(\eta_{\ep})\big)\|_{\infty}\\&
\leq \sum_{n\geq 0}\|(F_\ep^0(\eta_{\ep}))^n(f_{\ep}^{out}(\eta_{\ep})-h_{\ep}^{out}(\eta_{\ep}))\|_{\infty}\\&+
\sum_{n\geq 0}\|\big((F_\ep^0(\eta_{\ep}))^n  -(S_\ep^0(\eta_{\ep}))^n\big) h_{\ep}^{out}(\eta_{\ep})\|_{\infty}.
\end{split}
\end{equation}
By \eqref{F0bounded}, using Proposition \ref{th:propCIN}, the first sum on the right hand side of \eqref{tbyt} is bounded by
$$\frac{1}{1-\alpha'}\|f_{\ep}^{out}(\eta_{\ep})-h_{\ep}^{out}(\eta_{\ep})\|_{\infty}\leq C\ep^{\frac 1 2}\eta_{\ep}^5.$$
As regard to the second sum on the right hand side of \eqref{tbyt} we have
\begin{equation*}
\begin{split}
&\sum_{n\geq 0}\|\big((F_\ep^0(\eta_{\ep}))^n  -(S_\ep^0(\eta_{\ep}))^n\big) h_{\ep}^{out}(\eta_{\ep})\|_{\infty}\\&
\leq \sum_{n\geq 0}\sum_{k=0}^{n-1}\|(F_\ep^0(\eta_{\ep}))^{n-k-1} \big(F_\ep^0(\eta_{\ep})-S_\ep^0(\eta_{\ep})\big) (S_\ep^0(\eta_{\ep}))^k h_{\ep}^{out}(\eta_{\ep})\|_{\infty}\\&
\leq \sum_{k,\ell\geq 0}\|(F_\ep^0(\eta_{\ep}))^{\ell} \big(F_\ep^0(\eta_{\ep})-S_\ep^0(\eta_{\ep})\big) (S_\ep^0(\eta_{\ep}))^k h_{\ep}^{out}(\eta_{\ep})\|_{\infty}\\&
\leq C\,\|h_{\ep}^{out}(\eta_{\ep})\|_{\infty}\ep^{\frac 1 2}\eta_{\ep}^5,
 \end{split}
 \end{equation*}
by virtue of \eqref{S0bounded}, \eqref{F0bounded} and \eqref{eq:fepinhepin}. This concludes the proof of Proposition \ref{prop:fShS}.
 
\end{proof}
At this point the proof of Theorem \ref{th:MAIN1} follows from Propositions 
\ref{prop:hSrhoS} and \ref{prop:fShS}.

\begin{remark}\label{rk:asymp}
One could try to characterize $h_\ep^S$ and $f_\ep^S$ in terms of the long (macroscopic) time asymptotics of $h_{\ep}(t)$ and $f_{\ep}(t)$. The trick of expressing both stationary states by means of Neumann series avoids the problem of controlling the convergence rates, as $t\to\infty$, with respect to the scale parameter $\ep$.
\end{remark}

We conclude by proving Theorem \ref{th:MAIN2} which actually is a Corollary of the previous analysis.
\begin{proof} [Proof of Theorem \ref{th:MAIN2}]
By standard computations (see e.g. Section \ref{sec:HL2}) we have
$$
h_\ep^S=\varrho^S+\frac{1}{\eta_{\ep}}h^{(1)}+\frac{1}{\eta_{\ep}}R_{\eta_{\ep}},
$$
where 
$$
h^{(1)}(v)=\LL^{-1}(v\cdot\nabla_x\varrho^S)=\frac{\rho_2-\rho_1}{L}\LL^{-1}(v_1)
$$
and, as we shall see in Section \ref{sec:HL2}, $R_{\eta_{\ep}}=O(\frac{1}{\sqrt{\eta_{\ep}}})$ in $L^{2}((0,L)\times S_1).$
Therefore, since $\int_{S_1} v\varrho^S dv=0$, 
\begin{equation}\label{Fp}
\eta_{\ep}\int_{S_1} v h_\ep^S(x,v)dv=-D\nabla_x\varrho^S+O(\frac{1}{\sqrt{\eta_{\ep}}}),
\end{equation} 
where $D$ is given by \eqref{GK}. By Theorem \ref{th:MAIN1} the right hand side of \eqref{Fp} is close to $D\nabla_x\varrho_{\ep}^S$ in $\D'((0,L)\times S_1)$, where $\varrho_{\ep}^S$ is given by \eqref{def:mass}. On the other hand, by Proposition \ref{prop:fShS} and Assumption \ref{A1}, the left hand side of \eqref{Fp} is close in $L^{\infty}((0,L)\times S_1)$ to $J_{\ep}^S(x)$ defined in \eqref{def:j}. This concludes the proof of \eqref{FickL}. Moreover \eqref{Jlim} and \eqref{FICK} follow by \eqref{Fp}.
\end{proof}

\section{The Hilbert expansions}\label{sec4HILBERT}

\subsection{Proof of Proposition \ref{prop2}}\label{sec:HLinf}
Let $\h:\R^2\times S_1\times [0,T]\to \R^{+}$ be the solution of the problem \eqref{eq:Bfree} that we recall here for the reader's convenience
\begin{equation}\label{eq:Bfree2}
\left\{
\begin{array}{l}\vspace{4mm}
\big(\partial_t  +\eta_\ep\,v\cdot\nabla_x\big)\h=\eta_\ep^2\,  \LL \h\\
\h(x,v,0)=\varrho_0(x),
\end{array}
\right.
\end{equation}
where $\varrho_0$ is a smooth function of the variable $x$ only.  
We will prove that $\h$ converges to the solution of the heat equation
by using the Hilbert expansion technique (see e.g. \cite{EP} and \cite{CIP}), namely we assume that $\h$ has the following form
$$
\h(x,v,t)=h^{(0)}(x,t)+\sum_{k=1}^{+\infty} \left(\frac{1}{\eta_\ep}\right)^k\, h^{(k)}(x,v,t),
$$
where the coefficients $h^{(k)}$ are independent of $\eta_\ep$. The well known idea is to determine them recursively, by imposing that $\h$ is a solution of \eqref{eq:Bfree2}. 
Comparing terms of the same order we get
\begin{equation*}%\label{prima}
\begin{split}
& v\cdot\nabla_x h^{(0)}=\LL \, 
 h^{(1)}\\&
\partial_t\, h^{(k)}+v\cdot\nabla_x h^{(k+1)}=\LL \, 
 h^{(k+2)},\ \ \ \ \ \ k\geq 0.
 \end{split}
\end{equation*}
We require $h^{(0)}$ to satisfy the same initial condition as the whole solution $\h$, namely
\begin{equation*}%\label{in0}
h^{(0)}(x,0)=\varrho_0(x).
\end{equation*}
First we will show that each coefficient $h^{(k)}(t)\in L^\infty(\R^2\times S_1)$.
We discuss in detail the cases $k=0,1,2$. The same procedure can be iterated for any $k$. The determination of the other coefficients $h^{(k)}$ is standard and we do not discuss it further.
Then we will show that, in the truncated expansion
at order $\eta_\ep^{-2}$, namely
\begin{equation}\label{truncf}
\h(x,v,t)=h^{(0)}(x,t)+\frac 1 {\eta_\ep}h^{(1)}(x,v,t)+\frac 1 {\eta_\ep^2}h^{(2)}(x,v,t)
+\frac 1 {\eta_\ep}R_{\eta_\ep}(x,v,t),
\end{equation}
the remainder $R_{\eta_\ep}$ is uniformly bounded in $L^\infty$. Therefore $\h$ converges to $h^{(0)}$ in $L^\infty$ for $\eta_\ep\to\infty$. 

In order to prove that $h^{(k)}(t)\in L^\infty(\R^2\times S_1)$ we need the following Lemma.
\begin{lemma}\label{lem:Linfty}
Let $\LL $ be the linear Boltzmann operator defined in \eqref{def:L_ve}.
Then for any $g\in L^\infty(S_1)$ such that $\displaystyle \int_{S_1} dv\, g(v)=0$
\begin{equation}\label{LinversoBOUNDEDinLinfty}
|| \LL ^{-1}g||_\infty\leq C||g||_\infty, 
\end{equation}
with $C>0$.
\end{lemma}
\begin{proof}
We want to solve the equation
$\displaystyle
\LL h=g,
$
with $\displaystyle \int_{S_1} dv\, g(v)=0$. The operator $\LL $ can be written as $\LL =2\mu(K- I)$, where  
$$(Kf)(v):=
\frac{
1
}{2}
\int_{-1}^{1} d\rho\, f(v')
$$ 
is self-adjoint in $L^2(S_1)$.
Therefore
$$
h=-\frac{g}{2\mu} +Kh
$$
and, by iterating,
$$
h=-\frac{g}{2\mu} -\frac{Kg}{2\mu}-\dots -\frac{K^ng}{2\mu}+K^{n+1}h,\quad \forall\,n\geq 0.
$$
Then $\LL ^{-1}$ can be formally defined through the Neumann series
$$
h=\LL ^{-1}g:=-\frac{1}{2\mu}\sum_{n=0}^\infty K^n g.
$$
In order to prove that the series converges we need to show that
\begin{eqnarray}\label{Kbounded}
&& || Kg||_\infty\leq \beta\, ||g||_\infty,\ \ \ \beta<1,\\
\label{Kpreserva}
&& \int_{S_1} dv\, (Kg)(v)=0,
\end{eqnarray}
for any $g\in L^{\infty}(S_1)$
such that $\displaystyle \int_{S_1} dv\, g(v)=0$.
Indeed \eqref{Kbounded} and \eqref{Kpreserva} imply 
$$
|| \LL ^{-1}g||_\infty\leq \frac{1}{2\mu(1-\beta)}\ ||g||_\infty.
$$

The self-adjointness of $K$ and the fact that $K\text{1}=1$ imply (\ref{Kpreserva}).

We focus on the proof of (\ref{Kbounded}). For any given $v$, fix a reference system in such a way that $v=(-\cos\zeta, -\sin \zeta)$, with $ \zeta\in [-\pi, \pi)$ (see Figure \ref{fig:opK}). Then for every bounded function $g$ with zero average we have
\begin{equation*}
\begin{split}
(Kg)(v)&=
\frac{
1
}{2}\int_{-\frac{\pi}{2}}^{\frac{\pi}{2}} d\alpha\ \frac{d\rho}{d\alpha}\ g(\cos(\zeta +2\alpha), \sin(\zeta+2\alpha))\\&
=\frac{
1
}{2}\int_{-\frac{\pi}{2}}^{\frac{\pi}{2}} d\alpha\ \cos\alpha\  g(\cos(\zeta+2\alpha), \sin(\zeta+2\alpha)),
\end{split}
\end{equation*}
where we used that $\rho=\sin\alpha$.
Observe that for any $ \gamma\in [-\pi, \pi)$
\begin{eqnarray*}
\int_{-\frac{\pi}{2}}^{\frac{\pi}{2}} d\alpha\   g(\cos(\gamma+2\alpha), \sin(\gamma+2\alpha))=\frac{1}{2}\int_{-\pi}^{\pi} d\alpha\   g(\cos\alpha, \sin\alpha)=0.
\end{eqnarray*}
Then we can write
\begin{equation*}
(Kg)(v)=\frac{
1
}{2}\int_{-\frac{\pi}{2}}^{\frac{\pi}{2}} d\alpha\ (\cos\alpha-1)\  g(\cos(\zeta+2\alpha), \sin(\zeta+2\alpha)),
\end{equation*}
which implies
$$
|(Kg)(v)|\leq ||g||_\infty \ \frac{
1
}{2} \int_{-\frac{\pi}{2}}^{\frac{\pi}{2}} d\alpha\ (1-\cos\alpha)=
\beta\ ||g||_\infty,\ \  \ \ \beta<1.
$$
\end{proof}

Next we consider the first two equations arising from the Hilbert expansion, namely \vspace{1.5mm}

\begin{itemize} \vspace{1mm}
\item[(i)]  $\,v\cdot\nabla_x h^{(0)}=\LL \, h^{(1)}$,  \vspace{1mm}
\item[(ii)]  $\,\partial_t h^{(0)}+v\cdot\nabla_x h^{(1)}=\LL \, 
 h^{(2)}$.
 \end{itemize}
 \vspace{1.5mm}
We remind that the linear Boltzmann operator $\LL $ on $L^2(S_1)$ is selfadjoint and has the form $\LL =2\mu(K- I)$ where $K$ is a compact operator.
Therefore, by the Fredholm alternative, equation (i) has a solution if and only if the left hand side belongs to $(Ker \LL )^{\bot}$. Since the null space of $\LL $ is constituted by the constant functions, it follows that
$$
(Ker \LL )^{\bot}=\{g\in L^2(S_1):\ \int_{S_1} g(v)\, dv=0\}
$$
and, in order to solve equation (i), we have to show that $v\cdot\nabla_x h^{(0)}\in (Ker \LL )^{\bot}$. 
This follows by the fact that $v\cdot \nabla_x h^{(0)}$ is an odd function of $v$.
Then we can invert the operator $\LL $ and set 
\begin{equation}\label{h1f}
h^{(1)}(x,v,t)=\LL ^{-1}(v\cdot\nabla_x h^{(0)}(x,t)) + \xi^{(1)}(x,t),\ \ \ \ \ \ \ \ \ \ \ 
\end{equation}
where $\xi^{(1)}(x,t)$ belongs to the kernel of the operator $\LL $. On the other hand, since $\LL ^{-1}$ preserves the parity (see e.g. \cite{EP}), $\LL ^{-1}(v\cdot\nabla_x h^{(0)})$ is an odd function of the velocity.

We integrate equation (ii) with respect to the uniform measure on $S_1$. 
Since
$\int_{S_{1}} dv\,\LL \, h^{(2)}=0$, using equation (\ref{h1f}),
we obtain
\begin{equation*}
\partial_t \, 
h^{(0)}+\, 
\frac {1}{2\pi}\int_{S_{1}}dv\, v\cdot \nabla_x\big(
\LL ^{-1}v\cdot\nabla_x h^{(0)}\big)=0.\ \ \ \ \ \ \ \ \ \
\end{equation*}
Notice that the term $\xi^{(1)}(x,t)$ gives no contribution since $\int_{S_{1}}dv\, v\cdot \nabla_x \xi^{(1)}(x,t)=0$. We define the $2\times 2$ matrix  
$D_{ij}=\frac {1}{2\pi}\int_{S_{1}} dv\, v_i (-\LL )^{-1}v_j$
and we observe that $D_{ij}=0$ for $i\neq j$ and
$D_{11}=D_{22}=D>0$, where
$$
D=
\frac {1}{4\pi}\int_{S_{1}} dv\,v\cdot\big(-\LL   \big)^{-1}v. 
$$
Therefore $h^{(0)}$ satisfies the heat equation
\begin{equation}\label{eq:h0f}
\left\{
\begin{array}{l}\vspace{4mm}
\partial_t \, 
 h^{(0)}-D\Delta_x h^{(0)}=0,\\
 h^{(0)}(x,0)=\varrho_0(x).
 \end{array}\right.
\end{equation}
In particular $h^{(0)}(t)\in L^{\infty}(\R^2\times S_1)$ for any $t\geq 0$. 

Let us consider equation (ii). By integrating with respect 
to the uniform measure on $S_1$ the left hand side vanishes, due to equation \eqref{eq:h0f}. Therefore we can invert the operator $\LL $ to obtain
\begin{eqnarray}\label{h2f}
h^{(2)}(x,v,t)&= & \LL ^{-1}\big(\partial_t h^{(0)}(x,t) +v\cdot \nabla_x(\LL ^{-1}(v\cdot\nabla_x)h^{(0)}(x,t))  \big)+\nonumber\\
&+&\LL ^{-1}\big( v\cdot \nabla_x\xi^{(1)}(x,t)\big)+\xi^{(2)}(x,t),
\end{eqnarray}
where $\xi^{(2)}(x,t)$ belongs to the kernel of the operator $\LL $. 

The equation for $h^{(1)}$ reads
\begin{equation}\label{h3f}
\partial_t \, 
h^{(1)}+v\cdot \nabla_x h^{(2)}(x,v,t)=\LL h^{(3)}.
\end{equation}
Therefore, integrating with respect to the uniform measure on $S_1$, using \eqref{h2f}, we get the following closed equation for $\xi^{(1)}(x,t)$
\begin{equation}\label{rof}
\partial_t \, 
 \xi^{(1)}-D\Delta_x \xi^{(1)}=0.
\end{equation}
Since there are no restrictions on the initial condition, we make the simplest choice
$\xi^{(1)}(x,0)=0.$
Therefore
$\xi^{(1)}(x,t)=0$ for any $t\geq 0$ and
hence we have the following expression for $h^{(1)}$ 
\begin{equation*}%\label{h1fin}
h^{(1)}(x,v,t)=\LL ^{-1}(v\cdot\nabla_x h^{(0)}(x,t)).
\end{equation*}
Thanks to Lemma \ref{lem:Linfty} and the smoothness of $h^{(0)}$ we have
\begin{equation*}
\sup_{t\in [0,T]}
\|h^{(1)}( t)\|_\infty\leq \frac{1}{2\mu(1-\beta)} \sup_{t\in [0,T]}
\|\nabla_xh^{(0)}(t)\|_\infty <+\infty.
\end{equation*}

The expression for the second order coefficient $h^{(2)}$ now reads 
\begin{equation*}
\begin{split}
h^{(2)}(x,v,t)=h^{(2)}_{\bot}(x,v,t)+\xi^{(2)}(x,t),
\end{split}
\end{equation*}
where we set
\begin{equation*}
h^{(2)}_{\bot}(x,v,t)= \LL ^{-1}\big(\partial_t h^{(0)}(x,t) +v\cdot \nabla_x(\LL ^{-1}(v\cdot\nabla_x)h^{(0)}(x,t))  \big).
\end{equation*}
We observe that, since $h^{(0)}$ solves the heat equation \eqref{eq:h0f}, using Lemma \ref{lem:Linfty} it follows that $h^{(2)}_{\bot}\in L^\infty\big([0,T];L^\infty(\R^2\times S_1)\big)$. Moreover any spatial derivative of $h^{(2)}_{\bot}$ belongs to $L^\infty\big([0,T];L^\infty(\R^2\times S_1)\big)$ as well.

By using (\ref{rof}), the left hand side of (\ref{h3f}) belongs to $(Ker \LL )^{\bot}$. Therefore we can invert operator $\LL $ obtaining
\begin{equation*}\label{h3BB}
\begin{split}
h^{(3)}(x,v,t)&=\LL ^{-1}\big(\partial_t \, 
h^{(1)}+v\cdot \nabla_x h^{(2)}(x,v,t)\big) + \xi^{(3)}(x,t),\\&
=\LL ^{-1}\big(\partial_t \, 
\LL ^{-1}(v\cdot \nabla_x h^{(0)}(x,t))+v\cdot \nabla_x h^{(2)}(x,v,t)\big) + \xi^{(3)}(x,t),
\end{split}
\end{equation*}
where $\xi^{(3)}(x,t)\in \text{Ker}\,  \LL $.
The equation for $h^{(2)}$ reads
\begin{equation*}
\partial_t \, 
h^{(2)}+v\cdot \nabla_x h^{(3)}=\LL h^{(4)}.
\end{equation*}
Integrating with respect to the uniform measure on $S_1$ and using the above expressions for $h^{(3)}$ and $h^{(2)}$ we find the following equation
for $\xi^{(2)}(x,t)$
\begin{equation}\label{ro2}
\partial_t \, 
\xi^{(2)}+D\Delta_x\xi^{(2)} = S(x,t),
\end{equation}
where
\begin{eqnarray*}
S(x,t)&=&-\frac{1}{2\pi}\int_{S_1}dv\ v\cdot \nabla_x \, \LL ^{-1}\big(\partial_t \, 
\LL ^{-1}(v\cdot \nabla_x h^{(0)}(x,t))\big)\\
&&-
\frac{1}{2\pi}\int_{S_1}dv\ v\cdot \nabla_x\, \LL ^{-1}\left(v\cdot \nabla_x\, h^{(2)}_{\bot}(x,v,t)\right).
\end{eqnarray*}
We notice that $S\in L^\infty\big([0,T];L^\infty(\R^2)\big)$. As before we make the assumption $\xi^{(2)}(x,0)=0$, then $\xi^{(2)}\in L^\infty\big([0,T];L^\infty(\R^2)\big)$ and its spatial derivatives as well.

Now we consider the truncated expression \eqref{truncf}. 
The first three coefficients are uniformly bounded.
The remainder $R_{\eta_\ep}$ satisfies
\begin{equation}\label{resto}
 \,\big(\partial_t  +\eta_\ep\,v\cdot\nabla_x\big)R_{\eta_\ep}=\eta_\ep^2 \LL \,  R_{\eta_\ep}
-A_{\eta_\ep},
\end{equation}
with initial condition
\begin{equation*}
R_{\eta_\ep}(x,v,0)=:\bar{R}_{\eta_\ep}(x,v)=-h^{(1)}(x,v,0)-\frac{1}{\eta_\ep}\, h^{(2)}(x,v,0).
\end{equation*}
Here $A_{\eta_\ep}=\partial_t h^{(1)}+\frac 1 {\eta_\ep} \partial_t\, h^{(2)}
+v\cdot\nabla_x h^{(2)}$, then $A_{\eta_\ep}\in L^\infty\big([0,T];L^\infty(\R^2\times S_1)\big)$. Note that the smoothness hypothesis on $\varrho_0$ ensures that $\bar{R}_{\eta_\ep}\in L^{\infty}$.

We denote by $S_{\eta_\ep}(t)$ the semigroup associated to the  generator 
$-\eta_\ep\big(v\cdot \nabla_x - \eta_\ep\LL \big)$. By equation (\ref{resto})
we get
$$
R_{\eta_\ep}(t)=S_{\eta_\ep}(t)R_{\eta_\ep}(0)+\int_0^t ds\ S_{\eta_\ep}(t-s)\, A_{\eta_\ep}(s).
$$
By the usual series expansion for $S_{\eta_\ep}(t)$ we obtain
\begin{equation*}
\begin{split}
R_{\eta_\ep}(x,v,t)&=\sum_{N\geq 0} e^{-2\mu {\eta_\ep}^2 t}\left(\mu{\eta_\ep}\right)^{N}\int_{0}^{{\eta_\ep} t}dt_1\dots\int_{0}^{t_{N-1}}dt_N\\&
\int_{-1}^{1}d\rho_1\dots\int_{-1}^{1}d\rho_N\, \bar R_{\eta_\ep}(\gamma^{-{\eta_\ep} t}(x,v))+\\&
+\int_0^t ds\sum_{N\geq 0} e^{-2\mu{\eta_\ep}^2 (t-s)}\left(\mu{\eta_\ep}\right)^{N}\int_{0}^{{\eta_\ep}(t-s)}dt_1\dots\int_{0}^{t_{N-1}}dt_N\\&
\int_{-1}^{1}d\rho_1\dots\int_{-1}^{1}d\rho_N\, A_{\eta_\ep}(\gamma^{-{\eta_\ep}(t-s)}(x,v),s).
\end{split}
\end{equation*}
Therefore
\begin{equation*}
\sup_{t\in [0,T]}\|R_{\eta_\ep}( t)\|_\infty\leq \|\bar R_{\eta_\ep}\|_\infty+T \sup_{t\in [0,T]}\|A_{\eta_\ep}( t)\|_\infty\leq C <+\infty.
\end{equation*}

$\hfill\Box$

%%%%%%%%%%%%%%%%%%%%%%%%%%%%%%%%%%%%%%%%%%%%%%%%%%%%%%%%%%%%%%%%%%%%%%%%%%%%%%%%%%%%%%%%%%%%%%%
\subsection{Proof of Proposition \ref{prop:hSrhoS}}\label{sec:HL2}

The proof  makes use of the Hilbert expansion in $L^2$ (see e.g. \cite{EP} and \cite{CIP}).
Indeed we follow  the same strategy of the previous subsection.
Let $h_\ep^S$ be the solution of the following equation
\begin{equation*}
\left\{\begin{array}{ll}
v_1\partial_{x_1}h_\ep^{S}(x_1,v)=\eta_\ep\, \LL h_\ep^{S}(x_1,v),
&\vspace{0.2cm}\\
h_\ep^{S}(x_1,v)=\rho_1,\ \ \ \ \ x_1=0, \quad v_1>0,
&\vspace{0.2cm}\\
h_\ep^{S}(x_1,v)=\rho_2,\ \ \ \ \ x_1=L,\quad v_1< 0.
&
\end{array}\right.
\end{equation*}
We assume that $h_\ep^S$ has the following form
$$
h_\ep^S(x_1,v)=h^{(0)}(x_1)+\sum_{k=1}^{+\infty} \left(\frac{1}{\eta_\ep}\right)^k\, h^{(k)}(x_1,v).
$$
We require 
$h^{(0)}$ to satisfy the same boundary conditions as the whole solution $h_\ep^S$, namely
\begin{equation}\label{BC}
\left\{
 \begin{array}{l}\vspace{0.2cm}
h^{(0)}(x_1)=\rho_1,\ \ \ \ \ x_1=0, 
\vspace{0.2cm}\\
h^{(0)}(x_1)=\rho_2,\ \ \ \ \ x_1=L. 
\end{array} \right.
\end{equation}

Comparing terms of the same order we get
\begin{equation*}
\,v_1\partial_{x_1} h^{(k)}=\LL \, 
 h^{(k+1)},\ \ \ \ \ \ k\geq 0.
\end{equation*}

The first two equations read
\begin{itemize} \vspace{1mm}
\item[(i)]  $\,v_1\partial_{x_1} h^{(0)}=\LL \, h^{(1)}$,  \vspace{1mm}
\item[(ii)]  $\,v_1\partial_{x_1} h^{(1)}=\LL \, 
 h^{(2)}$,
 \end{itemize}
 \vspace{1.5mm}
which have a solution if and only if the left hand side belongs to $(Ker \LL )^{\bot}=\{g\in L^2(S_1):\ \int_{S_1} g(v)\, dv=0\}$. 
Since $v_1\partial_{x_1} h^{(0)}$ is an odd function of $v$
 we can invert the operator $\LL $ and set 
\begin{equation}\label{h1}
h^{(1)}(x_1,v)=\LL ^{-1}(v_1\partial_{x_1} h^{(0)}) + \xi^{(1)}(x_1),\ \ \ \ \ \ \ \ \ \ \ 
\end{equation}
where $\xi^{(1)}\in Ker \LL  $. We integrate equation (ii) with respect to the uniform measure on $S_1$. 
Observing that 
$\int_{S_{1}} dv\,\LL \, h^{(2)}=0$, by   \eqref{h1}
we obtain
\begin{equation*} 
\left(\int_{S_{1}}dv\, v_1
\LL ^{-1}v_1\right)\partial^2_{x_1} h^{(0)}=0,
\end{equation*}
with the boundary conditions \eqref{BC}. Therefore
\vspace{0.2mm}
\begin{equation*}
h^{(0)}(x_1)=\frac{\rho_1(L-x_1)+\rho_2x_1}{L}.
\end{equation*}
\vspace{1mm}
With the same strategy as the previous subsection, one can prove that $\xi^{(1)}(x_1)\equiv 0$. Hence
\begin{equation}\label{eq:h1S}
h^{(1)}(x_1,v)=h^{(1)}(v_1)=\displaystyle \left(\frac{\rho_2-\rho_1}{L}\right)\LL ^{-1}(v_1).
\end{equation}
Moreover by equation (ii) we get
\begin{equation*}
\begin{split}
h^{(2)}(x_1,v)&= \LL ^{-1}\big(v_1\partial_{x_1} h^{(1)}(x_1,v)  \big)+\xi^{(2)}(x_1)\\&
=\xi^{(2)}(x_1),
\end{split}
\end{equation*}
where in the last step we used \eqref{eq:h1S}.
By iterating the same procedure of the previous subsection, since in this case the source term in \eqref{ro2} is zero, we have that $\xi^{(2)}$ satisfies $\partial^2_{x_1}\xi^{(2)} = 0$.
We choose zero boundary conditions so that $\xi^{(2)}(x_1)\equiv 0$.
Then 
$$
h^{(2)}(x_1,v)\equiv 0.
$$

We consider the truncated expansion
\begin{equation}
h^S_{\ep}=h^{(0)}+\frac 1 {\eta_\ep}h^{(1)}
+\frac 1 {\eta_\ep}R_{\eta_\ep}.
\end{equation}
The remainder $R_{\eta_\ep}$ satisfies
\begin{equation}\label{restostaz}
 v_1\partial_{x_1}\,R_{\eta_\ep}=\eta_\ep \LL \,  R_{\eta_\ep}.
\end{equation}
We required 
$h^{(0)}$ to satisfy the same boundary conditions as the whole solution $h_\ep^S$, then
the boundary conditions for $R_{\eta_\ep}$ read
\begin{equation*}
\left\{
\begin{array}{l}\vspace{0.2cm}
R_{\eta_\ep}(x_1,v)=\displaystyle -\left(\frac{\rho_2-\rho_1}{L}\right)\LL ^{-1}(v_1),\ \ \ \ \ x_1=0,  \quad v_1>0,
\vspace{0.2cm}\\
R_{\eta_\ep}(x_1,v)=\displaystyle -\left(\frac{\rho_2-\rho_1}{L}\right)\LL ^{-1}(v_1),\ \ \ \ \ x_1=L,\quad v_1<0.
\end{array}\right.
\end{equation*}
The unique solution of the above problem is
\begin{equation*}
\begin{split}
R_{\eta_\ep}(x_1,v)=&-\displaystyle e^{\frac{\eta_\ep}{v_1}x_1\LL }\left(\frac{\rho_2-\rho_1}{L}\right)\LL ^{-1}(v_1)\chi(v_1>0)\\&\,- e^{-\frac{\eta_\ep}{v_1}(L-x_1)\LL }\left(\frac{\rho_2-\rho_1}{L}\right)\LL ^{-1}(v_1)\chi(v_1<0).
\end{split}
\end{equation*}
By \eqref{restostaz} we get
\begin{equation*}
-\eta_\ep\, \big(R_{\eta_\ep},\,-\LL  R_{\eta_\ep} \big)= - b_{\eta_\ep},
\end{equation*}
where the boundary term $b_{\eta_\ep}$ is given by
\begin{equation*}
b_{\eta_\ep}=- \int_{v_1>0} dv\, v_1\left(\frac{\rho_2-\rho_1}{L}\right)^2\big(\LL ^{-1}(v_1)\big)^2\left[e^{\frac{\eta_\ep}{v_1}L\LL }-1\right].
\end{equation*}
We remark that $(\cdot,\cdot)$ and  $\|\cdot\|_2$ denote the scalar product and the norm in $L^2((0,L)\times S_1)$ respectively.
Observe that $b_{\eta_\ep}\geq 0$. 
Using the spectral gap of the operator $\LL $ we get 
\begin{equation}\label{b_eta}
- b_{\eta_\ep}=-\eta_\ep\, \big(R_{\eta_\ep},\,-\LL  R_{\eta_\ep} \big) \leq -\lambda \eta_\ep \|R_{\eta_\ep}\|_2^2,
\end{equation}
where $\lambda$ is the first positive eigenvalue of $-\LL $. 
Therefore we obtain
\begin{equation*}
\|R_{\eta_\ep}\|_2 \leq \frac{C}{\sqrt{\eta_{\ep}}}. 
\end{equation*}
Since the coefficients 
$h^{(1)}$ and $h^{(2)}$ are bounded, we have that $h_\ep^S$ converges to $h^{(0)}$ in $L^2((0,L)\times S_1)$ for $\eta_\ep\to\infty$. 

$\hfill\Box$

\section{The kinetic description}\label{sec5}
\subsection{The extension argument}
We remind that $h^{out}_\ep$ is the solution of the Boltzmann equation \eqref{eq:Boltz}, therefore it can be expressed as
\begin{equation}
\begin{split}
\label{formula5}
h^{out}_{\ep}(x,v,t)&=\sum_{n\geq 0}(\mu_{\ep}\ep)^{n}\int_{0}^{t}dt_1\dots\int_{0}^{t_{n-1}}dt_n\\&
\int_{-1}^{1}d\rho_1\dots\int_{-1}^{1}d\rho_n\   \chi(\tau<t_n)\, \chi(\tau>0)\, e^{-2\mu_{\ep}\ep (t-\tau)}f_B(\gamma^{-(t-\tau)}(x,v)),
\end{split}
\end{equation}
with $f_B(x,v)$ defined in \eqref{def:fB} 
and 
\begin{equation}
\gamma^{-(t-\tau)}(x,v)=(x-v(t-\tau-t_1)-v_1(t_1-t_2)\dots-v_nt_n,v_n).
\end{equation}

\begin{lemma}\label{equivalence1}
Let $h^{out}_{\ep}$ be the solution of the Boltzmann equation \eqref{eq:Boltz} defined in \eqref{formula5}. Then
\begin{equation}
\begin{split}
\label{formula5*}
h^{out}_{\ep}(x,v,t)&=\sum_{N\geq 0} e^{-2\mu_{\ep}\ep t}(\mu_{\ep}\ep)^{N}\int_{0}^{t}dt_1\dots\int_{0}^{t_{N-1}}dt_N\\&
\int_{-1}^{1}d\rho_1\dots\int_{-1}^{1}d\rho_N\, \chi(\tau>0)\, f_B(\gamma^{-(t-\tau)}(x,v)).
\end{split}
\end{equation}
\end{lemma}
The above identity follows from the fact that in the last term we added fictitious jumps, those in the time interval $(0, \tau)$ which do not affect $f_B(\gamma^{-(t-\tau)}(x,v))$ but allows us to remove the indicator function $\chi(t_n>\tau)$ replacing consequently the factor $e^{-2\mu_{\ep}\ep (t-\tau)}$ by the more handable factor $e^{-2\mu_{\ep}\ep t}$. 
In view of the particle interpretation it is convenient to think the trajectory $\gamma^{-s}$, $s\in(0,t)$ as extended outside $\Lambda$, see Figure \ref{fig:EsTr}.
%FIGURA
\begin{figure}[ht]
\centering
\includegraphics[scale= 0.5]{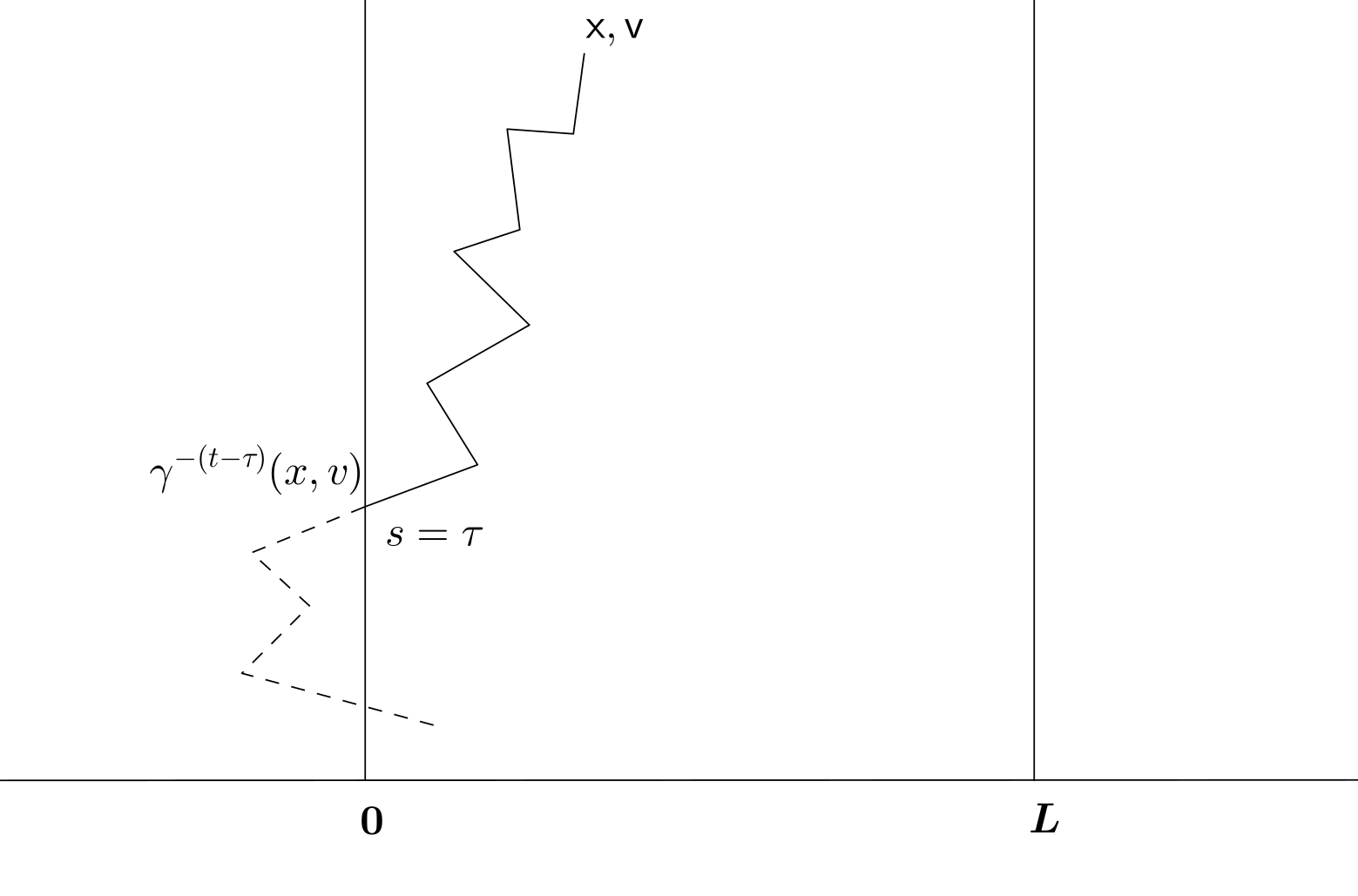}
\caption{Extension of the trajectory outside $\Lambda$}
\label{fig:EsTr}
\end{figure}
The dashed part of the trajectory is ininfluent for the evaluation of $h^{out}_{\ep}$.
\begin{proof}
Observe that for $\tau>0$, $\tau$ given,
\begin{equation*}
1=\sum_{m\geq 0}(\mu_{\ep}\ep)^{m}\int_{0}^{t}ds_1\dots\int_{0}^{s_{m-1}}ds_m\  \chi(s_1\leq \tau)
\int_{-1}^{1}d\xi_1\dots\int_{-1}^{1}d\xi_m\, e^{-2\mu_{\ep}\ep \tau}.
\end{equation*}
Using the previous identity we can express $h^{out}_{\ep}$ as 
\begin{equation*}
\begin{split}
h^{out}_{\ep}(x,v,t)&=\sum_{N\geq 0} e^{-2\mu_{\ep}\ep t}(\mu_{\ep}\ep)^{N}\int_{0}^{t}dt_1\dots\int_{0}^{t_{N-1}}dt_N \int_{-1}^{1}d\rho_1\dots\int_{-1}^{1}d\rho_N\, \\&
\left(\sum_{n=0}^{N} \chi(t_n>\tau) \chi(t_{n+1}\leq \tau)\right)\,  \chi(\tau>0)\, f_B(\gamma^{-(t-\tau)}(x,v)),
\end{split}
\end{equation*}
with the convention that $t_0=t$, $t_{N+1}=0$.
Since
\begin{equation*}
\left(\sum_{n=0}^{N} \chi(t_n>\tau) \chi(t_{n+1}\leq \tau)\right)=1,
\end{equation*}
we obtain the desired result.
\end{proof}

\subsection{Proof of Proposition \ref{th:propCIN}}     \label{proof:propCIN}
By \eqref{def:fepOUT} for $(x,v)\in\Lambda\times S_1$, $t>0$ we have
 \begin{equation*}
f^{out}_{\ep}(x,v,t)=e^{-\mu_{\ep}|B_t(x)\cap\Lambda\setminus B_{\ep}(x)|}\sum_{q\geq 0}\frac{\mu_{\ep}^{q}}{q!}\int_{(B_t(x)\cap\Lambda\setminus B_{\ep}(x))^q}d\mbf{c}_{q}\, \chi(\tau>0)f_B(T^{-(t-\tau)}_{\mbf{c}_{q}}(x,v)).
\end{equation*}
Here $T^{-(t-\tau)}_{\mbf{c}_{q}}(x,v)$ is the flow associated to the initial datum $(x,v)$ for a given scatterers configuration $\mbf{c}_{q}$, $B_t(x)$ and $B_{\ep}(x)$ denote the disks centered in $x$ with radius $t$ and $\ep$ respectively. 

Since $f_B(T^{-(t-\tau)}_{\mbf{c}_{q}}(x,v))$ depends only on the scatterer configurations inside $\Lambda\cap B_t(x)$, 
we want to add fictitious scatterers outside $\Lambda$ which do not affect the value $f_B(T^{-(t-\tau)}_{\mbf{c}_{q}}(x,v))$ in the same spirit of Lemma \ref{equivalence1}. However there is a small difficulty because the scatterers located in the vertical strips $[-\ep,0]\times\R$ and $[L,L+\ep]\times\R$ actually can modify the value of $\tau$. For this reason we introduce 
\begin{equation*}
\begin{split}\label{brevef}
%\label{formula1*}
\breve{f}^{out}_{\ep}(x,v,t)=&\,e^{-\mu_{\ep}|B^{\ep}_t(x)|}\sum_{Q\geq 0}\frac{\mu_{\ep}^{Q}}{Q!}\int_{(B^{\ep}_t(x))^Q}d\mbf{c}_{Q}\, \chi(\tau>0)\\& 
\big(1-\chi_{\partial\Lambda}{(\mbf{c}_{Q})}\big)f_B(T^{-(t-\tau)}_{\mbf{c}_{Q}}(x,v)),
\end{split}
\end{equation*}
where 
\begin{equation}\label{def:chiLambda}
\begin{split}
\chi_{\partial\Lambda}(\mbf{c}_{Q})=\chi\{\mbf{c}_{Q}: \, \exists i=1,\dots Q\; \text{s.t.}\; c_i\in[-\ep,0]\times\R\cup[L,L+\ep]\times\R\\\,\text{and}\,|x_{\ep}(-s)-c_i|=\ep,\,s\in[0,t]\},
\end{split}
\end{equation} 
allows to have a consistency in the definition of the hitting time for the extended dynamics. 
Here $B^{\ep}_t(x):=B_t(x)\setminus B_{\ep}(x)$. We expect that the contribution due to the obstacles with centers in the vertical strips $[-\ep,0]\times\R$, $[L,L+\ep]\times\R$ influencing the trajectory is indeed negligible in the limit. This fact will be discussed later on (see Section \ref{pathological}).

Since $|B^{\ep}_t(x)\setminus \{[-\ep,0]\times\R\cup[L,L+\ep]\times\R\}|  \leq |B^{\ep}_t(x)|$, then $f^{out}_{\ep}\geq \breve{f}^{out}_{\ep}$.

We distinguish the obstacles of the configuration $\mbf{c}_{Q}=c_1\dots c_Q$ which, up to the time $t$, influence the motion, called internal obstacles, and the external ones. More precisely, $c_i$ is internal if
\begin{equation*}
\inf_{0\leq s\leq t}|x_{\ep}(-s)-c_i|=\ep,
\end{equation*}
while $c_i$ is external if 
\begin{equation*}
\inf_{0\leq s\leq t}|x_{\ep}(-s)-c_i|>\ep.
\end{equation*}
Here $(x_{\ep}(-s),v_{\ep}(-s))=T^{-s}_{\mbf{c}_{Q}}(x,v)$, $s\in [0,t]$. We observe that the characteristic function $\chi_{\partial\Lambda}$ depends only on  internal obstacles. Therefore, by integrating over the external obstacles we obtain
\begin{equation*}
%\label{formula2*}
\begin{split}
\breve{f}^{out}_{\ep}(x,v,t)=&\sum_{N\geq 0}\frac{\mu_{\ep}^{N}}{N!}\int_{B^{\ep}_t(x)^N}d\mbf{b}_{N}\, e^{-\mu_{\ep}|\T_{t}(\mbf{b}_{N})|}\, \chi(\tau>0)\\&\, 
\chi(\{\mbf{b}_{N}\;\text{internal}\})\big(1-\chi_{\partial\Lambda}{(\mbf{b}_{N})}\big)f_B(T^{-(t-\tau)}_{\mbf{b}_{N}}(x,v)),
\end{split}
\end{equation*}
where $\T_{t}(\mbf{b}_{N})$ is the tube 
\begin{equation*}
\T_{t}(\mbf{b}_{N})=\{y\in B^{\ep}_t(x)\;\text{s.t.}\;\exists s\in [0,t]\;\text{s.t.}\;|y-x_{\ep}(-s)|\leq\ep\}.
\end{equation*}
We define 
\begin{equation*}
\begin{split}
\tilde{f}^{out}_{\ep}(x,v,t)=& e^{-2\mu_{\ep}\ep t}\sum_{N\geq 0}\frac{\mu_{\ep}^{N}}{N!}\int_{B^{\ep}_t(x)^N}d\mbf{b}_{N}
\,\chi(\{\mbf{b}_{N}\;\text{internal}\})\\&\,
\ (1-\chi_{\partial\Lambda}{(\mbf{b}_{N})}\big)\,f_B(T^{-(t-\tau)}_{\mbf{b}_{N}}(x,v))\, \chi(\tau>0).
\end{split}
\end{equation*}
Since
$|\T_{t}(\mbf{b}_{N})|\leq 2\ep t$,
then
$
f^{out}_{\ep}\geq \breve{f}^{out}_{\ep}\geq \tilde{f}^{out}_{\ep}.
$

Note that, according to a classical argument introduced in \cite{G} (see also \cite{DP}, \cite{DR}), we remove from $\tilde{f}^{out}_{\ep}$ all the bad events, namely those untypical with respect to the Markov process described by $h^{out}_{\ep}$. Then we will show they are unlikely.

For any fixed initial condition $(x,v)$ we order the obstacles $b_1,\dots,b_N$ according to the scattering sequence. Let $\rho_i$ and $t_i$ be the impact parameter and the hitting time of the light particle with $\partial B_{\ep}(b_i)$ respectively.
Then we perform the following change of variables
\begin{equation}
\label{change var}
b_1,\dots,b_N\rightarrow \rho_1,t_1,\dots,\rho_N,t_N
\end{equation}
with 
\begin{equation*}
0\leq t_N<t_{N-1}<\dots<t_1\leq t.
\end{equation*}
Conversely, fixed the impact parameters $\{\rho_i\}$ and the hitting times $\{t_i\}$ we construct the centers of the obstacles $b_i=b(\rho_i,t_i)$. By performing the backward scattering we construct a trajectory $\gamma^{-s}(x,v):=(\xi_{\ep}(-s),\omega_{\ep}(-s)),\ s\in[0,t],$
where 
\begin{equation}\label{limiting trajectory}
\left\{\begin{array}{ll}
\xi_{\ep}(-t)=x-v(t-t_1)-v_1(t_1-t_2)\dots-v_Nt_N&\\
\omega_{\ep}(-t)=v_N.&
\end{array}\right.
\end{equation}
Here $v_1,\dots,v_N$ are the incoming velocities. We remark that $\omega_{\ep}$ is an autonomous jump process and $\xi_{\ep}$
is an additive functional of $\omega_{\ep}$.

Observe that the map \eqref{change var} is one-to-one, and so $(\xi_{\ep}(-s),\omega_{\ep}(-s))=(x_{\ep}(-s),v_{\ep}(-s))$, only outside the following pathological situations.
\begin{itemize}
\item[i)] \textbf{Recollisions}.\\ There exists $b_i$ such that for $s\in(t_{j+1},t_{j})$, $j>i$, $\xi_{\ep}(-s)\in \partial B(b_i,\ep)$. 
\item[ii)] \textbf{Interferences}.\\ There exists $b_j$ such that $\xi_{\ep}(-s)\in B(b_j,\ep)$ for $s\in(t_{i+1},t_{i})$, $j>i$.
\end{itemize}
In order to skip such events we define
\begin{equation}
\begin{split}
\label{formula4}
\bar{f}^{out}_{\ep}(x,v,t)=& e^{-2\mu_{\ep}\ep t}\sum_{N\geq 0}\mu_{\ep}^{N}\int_{0}^{t}dt_1\dots\int_{0}^{t_{N-1}}dt_N\,\int_{-\ep}^{\ep}d\rho_1\dots\int_{-\ep}^{\ep}d\rho_N\, \\&
 \chi(\tau>0)\,(1-\chi_{\partial\Lambda})
(1-\chi_{rec})(1-\chi_{int})f_B(\gamma^{-(t-\tau)}(x,v)),
\end{split}
\end{equation}
where
\begin{equation}\label{def:chi}
\begin{split}
\chi_{rec}=&\chi(\{\mbf{b}_N\;\text{s.t.}\;\text{i)}\;\text{is}\;\text{realized}\})\\
\chi_{int}=&\chi(\{\mbf{b}_N\;\text{s.t.}\;\text{ii)}\;\text{is}\;\text{realized}\}).
\end{split}
\end{equation}
Observe that in \eqref{formula4} $\gamma^{-(t-\tau)}(x,v)=(x_{\ep}(-(t-\tau)),v_{\ep}(-(t-\tau)))$. Moreover $$\bar{f}^{out}_{\ep}\leq\tilde{ f}^{out}_{\ep}\leq \breve{f}^{out}_{\ep}\leq f^{out}_{\ep}.$$ 

Next we represent, thanks to Lemma \ref{equivalence1}, $h_{\ep}^{out}$, solution to equation \eqref{eq:Boltz}, as
\begin{equation}
\begin{split}
\label{formula5NN}
h^{out}_{\ep}(x,v,t)= & e^{-2\mu_{\ep}\ep t}\sum_{N\geq 0}\mu_{\ep}^{N}\int_{0}^{t}dt_1\dots\int_{0}^{t_{N-1}}dt_N \\&
\int_{-\ep}^{\ep}d\rho_1\dots\int_{-\ep}^{\ep}d\rho_N\, \chi(\tau>0)\, f_B(\gamma^{-(t-\tau)}(x,v)).
\end{split}
\end{equation}
Observe that 
\begin{equation}\label{eq:chi}
1-(1-\chi_{\partial\Lambda})(1-\chi_{rec})(1-\chi_{int}) \leq \chi_{\partial\Lambda}+\chi_{rec}+\chi_{int}.
\end{equation}
Then by \eqref{formula4} and \eqref{formula5NN} we obtain
\begin{equation}\label{est:ffi1}
|h^{out}_{\ep}(t)-\bar{f}^{out}_{\ep}(t)|\leq \ffi_1(\ep,t),
\end{equation}
with
\begin{equation}\label{def:ffi1}
\begin{split}
\ffi_1(\ep,t):=& \,\|f_B\|_{\infty}e^{-2\mu_{\ep}\ep t} \sum_{N\geq 0}(\mu_{\ep})^N\int_{0}^{t}dt_1\dots\int_{0}^{t_{N-1}}dt_N
\\&
\int_{-\ep}^{\ep}d\rho_1\dots\int_{-\ep}^{\ep}d\rho_N \, \{\chi_{\partial\Lambda}+\chi_{rec}+\chi_{int}\}.
\end{split}
\end{equation}

We state the following result. The proof is postponed to Section \ref{pathological}.
\begin{lemma}\label{th:prop}
Let $\ffi_1(\ep,t)$ be defined in \eqref{def:ffi1}. For any $t\in[0,T]$ we have
\begin{equation}\label{norm:ffi1}
\|\ffi_1(\ep,t)\|_{L^{\infty}}\leq C\varepsilon^{\frac{1}{2}}\, \eta_\ep^3\, t^2. 
\end{equation}
\end{lemma}
\vspace{4.5mm}

Let us estimate the difference $\left|f^{out}_\ep(t)-h^{out}_\ep(t)\right|$. By \eqref{est:ffi1} we have
\begin{equation}\label{monotonia}
\begin{split}
\left|f^{out}_\ep(t)-h^{out}_\ep(t)\right|& \leq \left|f^{out}_\ep(t)-\bar{f}^{out}_\ep(t)\right|+\left|\bar{f}^{out}_\ep(t)-h^{out}_\ep(t)\right|\\&
\leq \left|f^{out}_\ep(t)-\bar{f}^{out}_\ep(t)\right|+ \ffi_1(\ep,t).
\end{split}
\end{equation}
Since $\bar{f}^{out}_{\ep}\leq f^{out}_{\ep}$, the difference $f^{out}_\ep(t)-\bar{f}^{out}_\ep(t)$ is non negative and we can skip the absolute value. Moreover 
\begin{equation}
f^{out}_\ep(t)-\bar{f}^{out}_\ep(t)\leq \left(f^{out}_\ep(t)-\breve{f}^{out}_\ep(t)\right)+\left(\breve{f}^{out}_\ep(t)-\bar{f}^{out}_\ep(t)\right).
\end{equation}

Using the fact that the map \eqref{change var} is one-to-one outside the pathological sets we can write $\bar{f}^{out}_\ep$ in \eqref{formula4} as 
\begin{equation*}
\begin{split}
\bar{f}^{out}_\ep(t)&=e^{-2\mu_{\ep}\ep t}\sum_{N\geq 0}\frac{\mu_{\ep}^{N}}{N!}\int_{B^{\ep}_t(x)^N}d\mbf{b}_{N}\,
\chi(\{\mbf{b}_{N}\;\text{internal}\})\, \chi(\tau>0) \\&
(1-\chi_{\partial\Lambda})
(1-\chi_{rec})(1-\chi_{int})f_B(T^{-(t-\tau)}_{\mbf{b}_{N}}(x,v)). 
\end{split}
\end{equation*}
Hence
\begin{equation*}%\label{convergenza2}
\begin{split}
\breve{f}^{out}_\ep(t)-\bar{f}^{out}_\ep(t)=&\sum_{N\geq 0}\frac{\mu_{\ep}^{N}}{N!}\int_{B^{\ep}_t(x)^N}d\mbf{b}_{N}\, f_B(T^{-(t-\tau)}_{\mbf{b}_{N}}(x,v))\
\chi(\{\mbf{b}_{N}\;\text{internal}\})\, \chi(\tau>0)
\\&
\qquad (1-\chi_{\partial\Lambda}) \big(e^{-\mu_{\ep}|\T_{t}(\mbf{b}_{N})|}-e^{-2\, \mu_{\ep}\ep\, t}\,
(1-\chi_{rec})(1-\chi_{int})
\big)\\
\leq &\,|| f_B||_\infty\sum_{N\geq 0}\frac{\mu_{\ep}^{N}}{N!}\int_{B^{\ep}_t(x)^N}d\mbf{b}_{N}\
\chi(\{\mbf{b}_{N}\;\text{internal}\})\, 
\\&
\qquad %(1-\chi_{\partial\Lambda}) %lo maggioro con 1
\big(e^{-\mu_{\ep}|\T_{t}(\mbf{b}_{N})|}-e^{-2\, \mu_{\ep}\ep\, t}\,
(1-\chi_{rec})(1-\chi_{int})
\big).
\end{split}
\end{equation*}

By observing that
\begin{equation*}
\sum_{N\geq 0}\frac{\mu_{\ep}^{N}}{N!}\int_{B^{\ep}_t(x)^N}d\mbf{b}_{N}\,
\chi(\{\mbf{b}_{N}\;\text{internal}\})\,e^{-\mu_{\ep}|\T_{t}(\mbf{b}_{N})|}=1,
\end{equation*}
we obtain
\begin{equation*}
\begin{split}
\breve{f}^{out}_\ep(t)-\bar{f}^{out}_\ep(t)& \leq
 || f_B||_\infty 
 \left(1-e^{-2\, \mu_{\ep}\ep\, t}\right.\, \sum_{N\geq 0} \mu_{\ep}^{N}\int_{0}^{t}dt_1\dots\int_{0}^{t_{N-1}}dt_N
 \\& 
\left.\quad\  \int_{-\ep}^{\ep}d\rho_1\dots\int_{-\ep}^{\ep}d\rho_N\, 
%(1-\chi_{\partial\Lambda})
(1-\chi_{rec})(1-\chi_{int})\right).
\end{split}
\end{equation*}
Observe that 
\begin{equation*}
1-(1-\chi_{rec})(1-\chi_{int}) \leq \chi_{rec}+\chi_{int}.
\end{equation*}
Hence we get
\begin{equation}\label{alessia1}
\breve{f}^{out}_\ep(t)-\bar{f}^{out}_\ep(t)\leq \ffi_1(\ep,t), 
\end{equation}
with $\ffi_1$ defined in \eqref{def:ffi1}.

Now we consider $f^{out}_\ep(t)-\breve{f}^{out}_\ep(t)$. We observe that
\begin{equation*}
\begin{split}
%\label{formula1*}
f^{out}_{\ep}(x,v,t)=&\,e^{-\mu_{\ep}|B^{\ep}_t(x)\setminus \partial\Lambda^{\ep}|}\sum_{Q\geq 0}\frac{\mu_{\ep}^{Q}}{Q!}\int_{(B^{\ep}_t(x))^Q}d\mbf{c}_{Q}\, \chi(\tau>0)\\& 
\big(1-\chi_{\partial\Lambda}{(\mbf{c}_{Q})}\big)f_B(T^{-(t-\tau)}_{\mbf{c}_{Q}}(x,v)),
\end{split}
\end{equation*}
where $\partial\Lambda^{\ep}:= \big([-\ep,0]\cup [L,L+\ep]\big)\times \R$. By using the previous strategy one can prove
\begin{equation}\label{alessia2}
f^{out}_\ep(t)-\breve{f}^{out}_\ep(t)\leq  \ffi_1(\ep,t). 
\end{equation}
%Now we consider 
%\begin{equation*}
%f^{out}_\ep(t)-\breve{f}^{out}_\ep(t)\leq \ffi_1(\ep,t), 
%\end{equation*}
Therefore \eqref{monotonia}, \eqref{alessia1}, \eqref{alessia2} and \eqref{norm:ffi1} imply
\begin{equation*}
\|f^{out}_{\ep}(t)-h^{out}_{\ep}(t)\|_{\infty}\leq C\varepsilon^{\frac{1}{2}}\, \eta_\ep^3\, t^2.
\end{equation*}

\subsection{Proof of Lemma \ref{th:prop} (the control of the pathological sets)}\label{pathological}

For any measurable function $u$ of the process $(\xi_{\ep},\omega_{\ep})$ defined in \eqref{limiting trajectory} we set
\begin{equation*}
\begin{split}
\EE _{x,v}(u)&=e^{-2\mu_{\ep}\ep t} \sum_{N\geq 0}(\mu_{\ep})^N\int_{0}^{t}dt_1\dots\int_{0}^{t_{N-1}}dt_N\\&
\int_{-\ep}^{\ep}d\rho_1\dots\int_{-\ep}^{\ep}d\rho_N \,u(\xi_{\ep},\omega_{\ep}).
\end{split}
\end{equation*}
Then we realize that 
\begin{equation*}
\ffi_1(\ep,t)= \|f_B\|_{\infty}\,\EE_{x,v}[\chi_{\partial\Lambda}+\chi_{rec}+\chi_{int}]
\end{equation*}
and we estimate separately the events in \eqref{def:chiLambda} and \eqref{def:chi}.

We consider the interference event.  Let $t_i$ the first time the light particle hits the i-th scatterer, $v_i^{-}$ the incoming velocity and $v_i^{+}$ the outgoing velocity (for the backward trajectory). Moreover we fix the axis in such a way that $v_i^{+}$ is parallel to the $x$ axis. 
We have
\begin{equation*} 
\chi_{int}\leq \sum_{i=1}^{N}\sum_{j>i}\chi_{int}^{i,j},
\end{equation*}
where $\chi_{int}^{i,j}=1$ if the obstacle with center $b_j$ belongs to the tube spanned by $\xi_{\ep}(-s)$ for $s\in (t_{i+1},t_{i})$.
We denote by $\alpha$ the angle between $v_i^{+}$ and $v_{j-1}^{+}$.  We have two situations, when the velocity $v_{j-1}^{+}$ is transverse to $v_i^{+}$ (i.e. $\alpha>\ep^{\gamma}$ for a suitable positive $\gamma$) or when the velocity $v_{j-1}^{+}$ is almost parallel to $v_i^{+}$ (i.e. $\alpha\leq\ep^{\gamma}$). Then
\begin{equation}
\begin{split}
\label{Markov error int}
\EE _{x,v}[\chi_{int}]\leq
\EE _{x,v}\big[ \sum_{i=1}^{N}\sum_{j>i}\chi_{int}^{i,j}\chi(\alpha>\ep^{\gamma})\big]+\EE _{x,v}\big[ \sum_{i=1}^{N}\sum_{j>i}\chi_{int}^{i,j}\chi(\alpha\leq\ep^{\gamma})\big ].
\end{split}
\end{equation}
%FIGURA
\begin{figure}[ht]
\centering
\includegraphics[scale= 0.6]{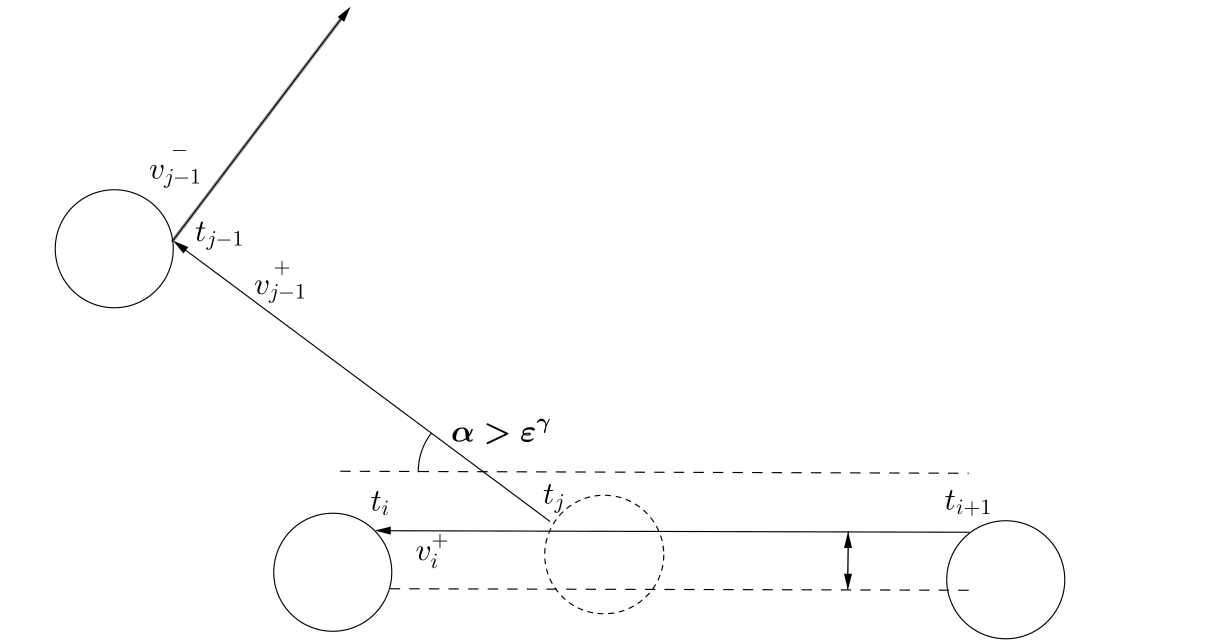}
\caption{Backward Interference-First case}\label{fig:1}
\end{figure}

We estimate separately the two contributions. To estimate the first term we fix all the variables $\{t_h\}_{h=1}^{N}$, $\{ \rho_h\}_{h=1}^{N}$ except $t_j$. By a simple geometrical argument we argue that the integral over $t_j$ is restricted over an interval of measure at most $C\ep^{1-\gamma}$. Hence we get
\begin{equation}
\label{Markov error int1}
\begin{split}
&\,\EE _{x,v}\big[ \sum_{i=1}^{N}\sum_{j>i}\chi_{int}^{i,j}\chi(\alpha>\ep^{\gamma})\big]\\&
\leq e^{-2\mu_{\ep}\ep t}  \sum_{N\geq 1}(N)^2(2\mu_{\ep}\ep)^N\frac{ t^{N-1}}{(N-1)!}C\ep^{1-\gamma}\\&
\leq e^{-2\mu_{\ep}\ep t}  (2\mu_{\ep}\ep)^{3}t^2\sum_{N\geq 3}(2\mu_{\ep}\ep)^{N-3}\frac{ t^{N-3}}{(N-3)!}C\ep^{1-\gamma}\\&
\leq C\ep^{1-\gamma}\eta_\ep^3 t^2.
\end{split}
\end{equation}
%FIGURA
\begin{figure}[ht]
\centering
\includegraphics[scale= 0.17]{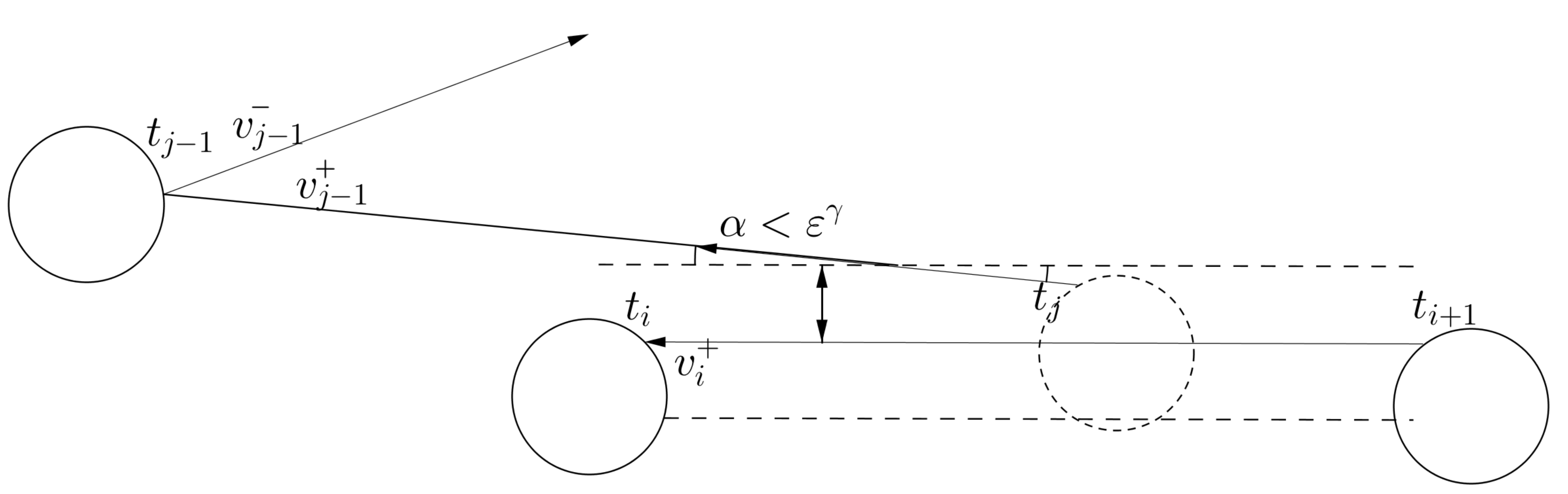}
\caption{Backward Interference-Second case}\label{fig:2}
\end{figure}

Concerning the second term in \eqref{Markov error int}, the condition $\alpha\leq\ep^{\gamma}$ implies that the $(j-1)$-th scattering angle $\theta_{j-1}$ can varies at most $\ep^{\gamma}$ (see Figure \ref{fig:2}). Then, fixing all 
the variables $\{t_h\}_{h=1}^{N}$, $\{ \rho_h\}_{h=1}^{N}$ except $\rho_{j-1}$, performing the change of variable $\rho_{j-1}\to \theta_{j-1}$ and recalling that the scattering cross section for a disk of unitary radius is given by $\frac{d\rho}{d\theta}=\frac 1 2 \sin\frac{\theta}{2}$ we obtain
\begin{equation}
\label{Markov error int2}
\begin{split}
\EE _{x,v}\big[ \sum_{i=1}^{N}\sum_{j>i}\chi_{int}^{i,j}\chi(\alpha\leq\ep^{\gamma})\big]
\leq & \,e^{-2\mu_{\ep}\ep t} \sum_{N\geq 1}(N)^2(2\mu_{\ep}\ep) ^{N}\frac{ t^{N}}{(N)!}C\ep^{\gamma}\\
\leq & \,C\ep^{\gamma}\eta_\ep^2 t^2.
\end{split}
\end{equation}
By choosing $\gamma=1/2$, from \eqref{Markov error int1} and \eqref{Markov error int2} we obtain
\begin{equation}
\label{Markov error int finale}
\EE _{x,v}[\chi_{int}]\leq  C\ep^{\frac 1 2}\eta_\ep^3t^2.
\end{equation}

Finally we consider the recollision event. We have
\begin{equation*} 
\chi_{rec}\leq \sum_{i=1}^{N}\sum_{j>i}\chi_{rec}^{i,j},
\end{equation*}
where  $\chi_{rec}^{i,j}=1$ if the $i$-th obstacle is recollided in the time interval $(t_{j},t_{j-1})$. Also in this case we have to take into account two possible situations, when 
$|b_i-b_{j-1}|> \ep^{\gamma}$ for a suitable positive $\gamma$ or when  
$|b_i-b_{j-1}|\leq \ep^{\gamma}$. Then
\begin{equation}
\begin{split}
\label{Markov error rec}
\EE _{x,v}[\chi_{rec}]
\leq &\, \EE _{x,v}\Big[ \sum_{i=1}^{N}\sum_{j>i}\chi_{rec}^{i,j}\chi\big(|b_i-b_{j-1}|> \ep^{\gamma} \big)\Big]\\&+\EE _{x,v}\Big[ \sum_{i=1}^{N}\sum_{j>i}\chi_{rec}^{i,j}\chi\big(|b_i-b_{j-1}|\leq \ep^{\gamma} \big)\Big].
\end{split}
\end{equation}

%FIGURA
\begin{figure}[ht]
\includegraphics[scale= 0.25]{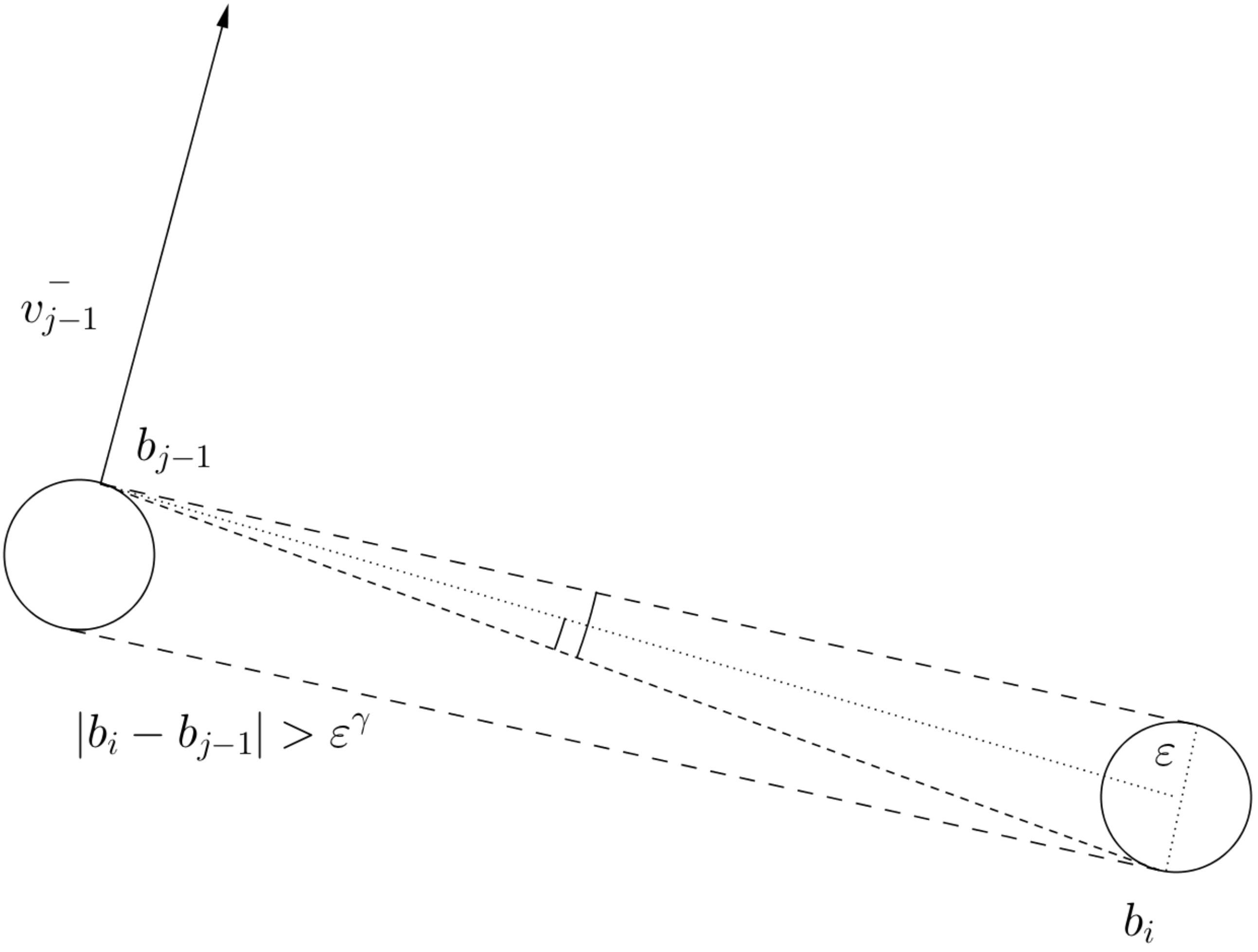}
\caption{Backward Recollision-First case}\label{fig:4}
\end{figure}

%%FIGURA
\begin{figure}[ht]
\includegraphics[scale= 0.2]{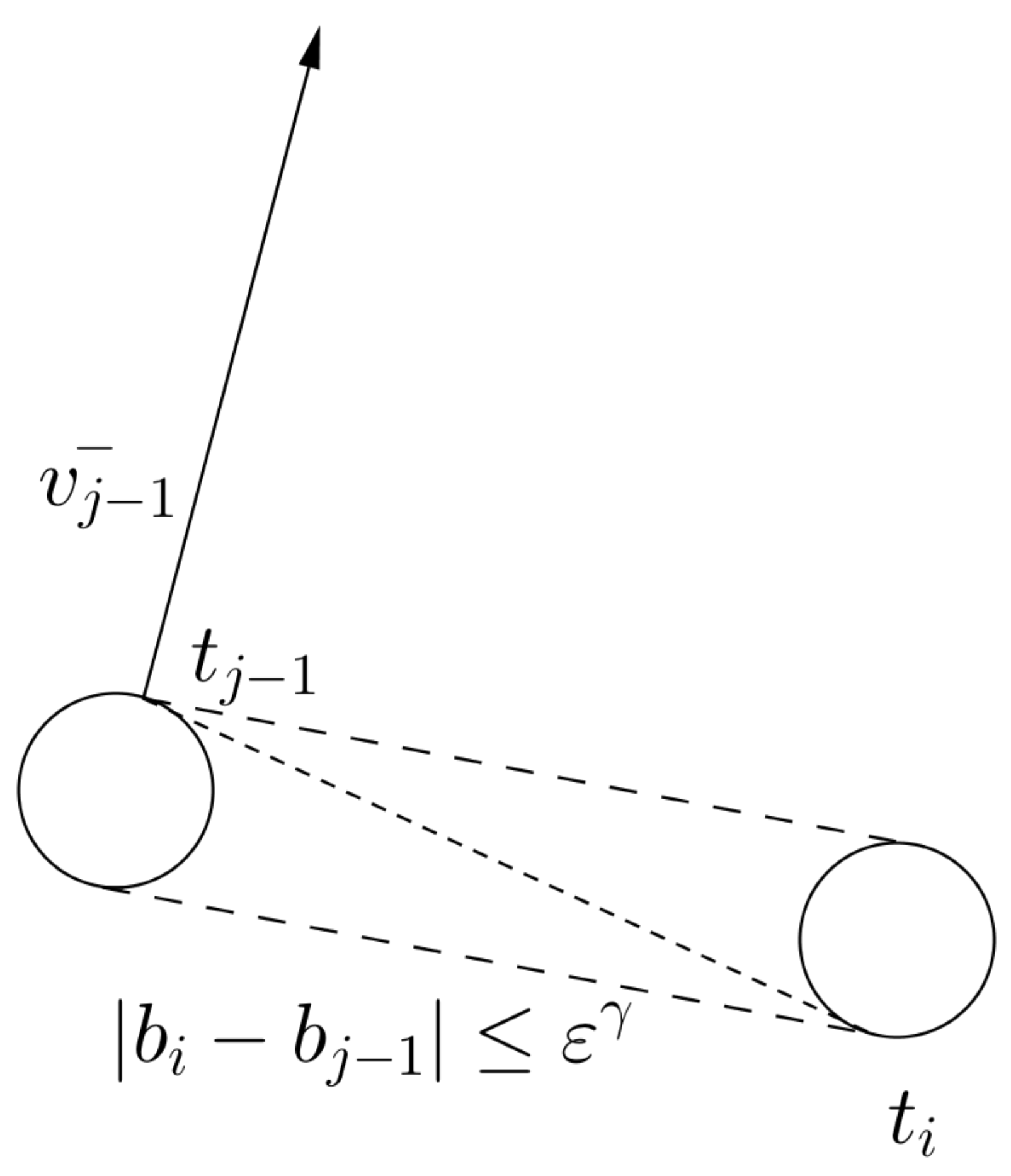}
\caption{Backward Recollision-Second case}\label{fig:3}
\end{figure}

We look at the first term. Using geometric arguments the condition $|b_i-b_{j-1}|> \ep^{\gamma}$ gives a bound for the $(j-1)$-th scattering angle $\theta_{j-1}$ (see Figure \ref{fig:4}). In particular it can varies at most $\ep/\ep^{\gamma}=\ep^{1-\gamma}$. Therefore, performing the change of variable $\rho_{j-1}\to \theta_{j-1}$ as before, we get
\begin{equation}
\begin{split}
\label{Markov error rec1}
&\EE _{x,v}\Big[ \sum_{i=1}^{N}\sum_{j>i}\chi_{rec}^{i,j}\chi\big(|b_i-b_{j-1}|> \ep^{\gamma} \big)\Big]\\&
\leq e^{-2\mu_{\ep}\ep t}  \sum_{N\geq 1}(N)^2(2\mu_{\ep}\ep) ^{N}\frac{ t^{N}}{(N)!}C\ep^{1-\gamma}\\&
\leq C\ep^{1-\gamma}\eta_\ep^2 t^2.
\end{split}
\end{equation}

If $|b_i-b_{j-1}|\leq \ep^{\gamma}$ a simple geometrical argument shows that the time interval $|t_{j-1}-t_j|$ is bounded by $\ep^{\gamma}$ (see Figure \ref{fig:3}). Hence,  following the same strategy as in \eqref{Markov error int1}, we obtain
\begin{equation}
\begin{split}
\label{Markov error rec2}
&\EE _{x,v}\Big[ \sum_{i=1}^{N}\sum_{j>i}\chi_{rec}^{i,j}\chi\big(|b_i-b_{j-1}|\leq \ep^{\gamma} \big)\big]\\&
\leq e^{-2\mu_{\ep}\ep t} \sum_{N\geq 1}(N)^2(2\mu_{\ep}\ep)^N\frac{ t^{N-1}}{(N-1)!}C\ep^{\gamma}\\&
\leq C\ep^{\gamma}\eta_\ep^3 t^2.
\end{split}
\end{equation}
As before we choose $\gamma=1/2$. Then from \eqref{Markov error rec1} and \eqref{Markov error rec2} we obtain
\begin{equation}
\label{Markov error rec finale}
\EE _{x,v}[\chi_{rec}]\leq  C\ep^{\frac 1 2}\eta_\ep^3t^2.
\end{equation}

%FIGURA
\begin{figure}[ht]
\includegraphics[scale= 0.18]{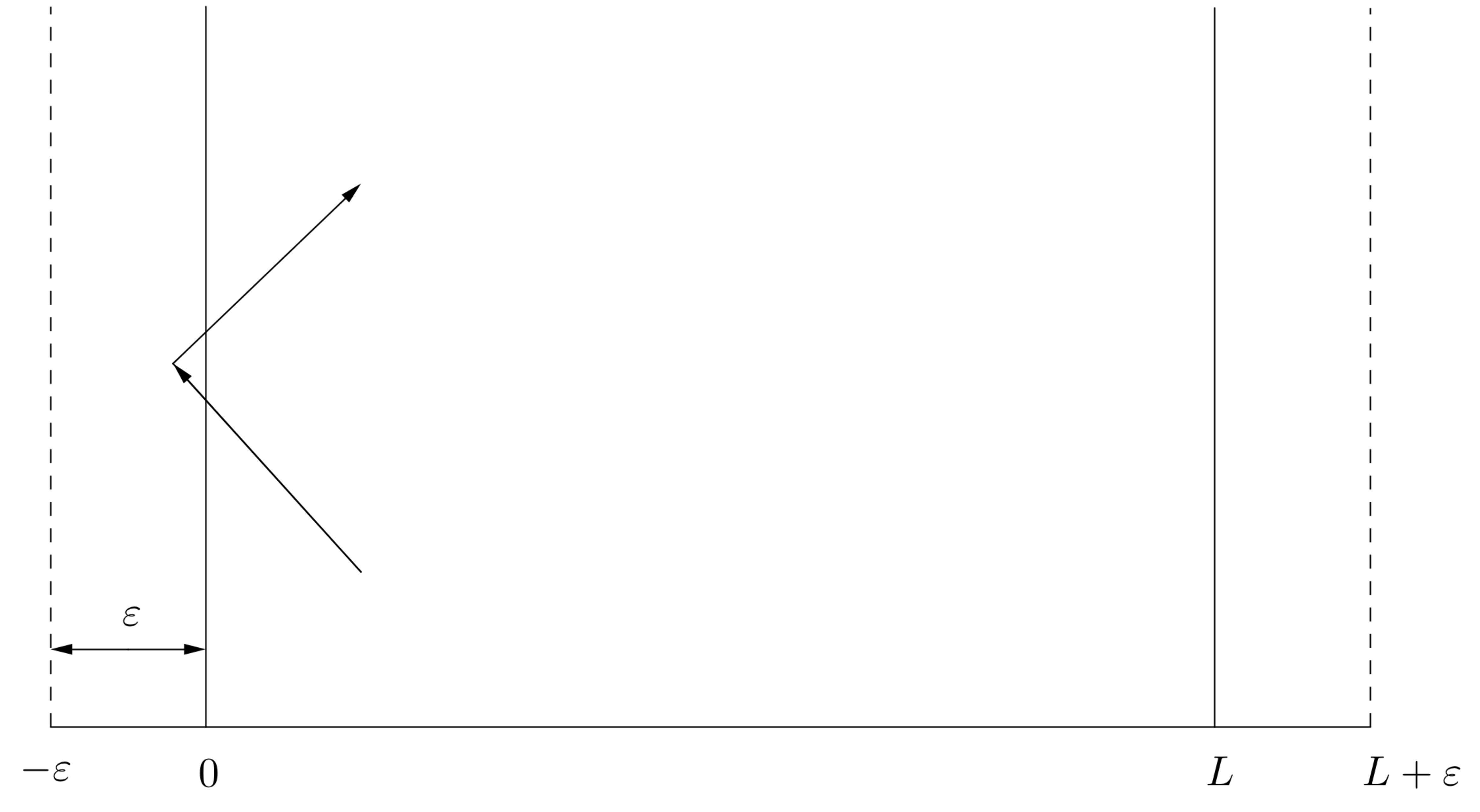}
\caption{$\Lambda\cup \{[-\ep,0]\times\R\cup[L,L+\ep]\times\R\}$}\label{fig:L}
\end{figure}
\vspace{0.4mm}

We now consider the expectation value for $(1-\chi_{\partial\Lambda})$, with $\chi_{\partial\Lambda}$ defined in \eqref{def:chiLambda}. Observe that $\chi_{\partial\Lambda}=1$ implies that $\xi_{\ep}(-(t-t_j))\in\Lambda^\textit{c}$ and $d(\xi_{\ep}(-(t-t_j)),\partial \Lambda)\leq\ep$ for some $j=1,\dots,N$.
As we can see in Figure \ref{fig:L}, 
by the same argument used to estimate the interference events in \eqref{Markov error int1} and \eqref{Markov error int2} we obtain
\begin{equation}
\label{Markov bd error}
\EE _{x,v}[\chi_{\partial\Lambda}]\leq  C\ep^{\frac 1 2}\eta_\ep^3t^2. 
\end{equation}

By estimates \eqref{Markov error int finale}, \eqref{Markov error rec finale} and \eqref{Markov bd error} we obtain 
$$
\|\ffi_1(\ep,t)\|_{\infty}\leq  C\ep^{\frac 1 2}\eta_\ep^3t^2,
$$
for some $C>0$.

\subsection{Proof of Proposition \ref{prop:fepinhepin}}
The proof follows the same strategy of the proof of Proposition \ref{th:propCIN}. Actually it is easier since it does not require the extension trick, but it follows directly by the recollision and interference estimates.

\vspace{15mm}
\indent\textbf{Acknowledgments.}\\ F. Pezzotti is partially supported by the Italian Project FIRB 2012 "Dispersive dynamics: Fourier Analysis and Variational Methods"
\vspace{10mm}

\setcounter{equation}{0}    
\def\theequation{A.\arabic{equation}}

\end{document}